\documentclass{article}

\usepackage[utf8]{inputenc}
\usepackage[T1]{fontenc}
\usepackage[a4paper,left=3.3cm,right=3.3cm,bottom=2.2cm,top=2cm,foot=1.3cm,includeheadfoot=true]{geometry}
\usepackage{todonotes}

\usepackage{lmodern}
\usepackage{mathrsfs}
\usepackage{amsmath}
\usepackage{amsfonts}
\usepackage{amssymb}
\usepackage{amsmath, amsthm}
\usepackage{framed}
\usepackage[boxed]{algorithm2e}
\usepackage[hidelinks]{hyperref}
\usepackage{multirow,bigdelim}
\usepackage[babel=true]{csquotes}
\usepackage{enumitem}
\usepackage{pgfplots}
\usepackage{adjustbox}

\DeclareMathOperator{\Z}{\mathbb{Z}}
\DeclareMathOperator{\C}{\mathcal{C}}

\DeclareMathOperator{\GCD}{GCD}
\DeclareMathOperator{\XCRT}{XCRT}
\DeclareMathOperator{\LCM}{LCM}
\DeclareMathOperator{\Frac}{Frac}

\DeclareMathOperator{\red}{red}
\DeclareMathOperator{\Resultant}{Resultant}
\DeclareMathOperator{\FirstSubRes}{FirstSubRes}

\DeclareMathOperator{\Mcomp}{M}
\DeclareMathOperator{\Gal}{Gal}
\DeclareMathOperator{\val}{val}

\renewcommand{\P}{\mathbb P}

\newcounter{countproblem}

\date{}

\newtheorem{theorem}{Theorem}[section]
\newtheorem{lemma}[theorem]{Lemma}
\newtheorem{proposition}[theorem]{Proposition}

\begin{document}
\pagestyle{empty}
\begin{center}
  \Large\textbf{Erratum}\\[0.2cm]
  \large\today
\end{center}

\bigskip\bigskip\bigskip

The formula on page 4, line 14, should read
$$d = \begin{cases}
  \lfloor (\deg(D_+)+r)/\deg(\C)+(\deg(\C)-1)/2\rfloor\text{ if
  }\binom{\deg(\C)+1}2\leq \deg(D_+) + r\\
  \lfloor (\sqrt{1+8(\deg(D_+) + r)}-1)/2\rfloor\text{\quad\quad\quad\quad~ otherwise.}
\end{cases}$$

\pagebreak
\setcounter{page}{1}
\pagestyle{plain}

\title{A Fast Randomized Geometric Algorithm for Computing Riemann-Roch Spaces}
\author{Aude Le Gluher and Pierre-Jean Spaenlehauer\\[0.2cm]
{\normalsize Université de Lorraine, CNRS, Inria}}

\maketitle

\begin{abstract}
  We propose a probabilistic variant of Brill-Noether's algorithm for computing a basis of
  the
  Riemann-Roch space $L(D)$ associated to a divisor $D$ on a projective nodal plane curve
  $\C$ over a sufficiently large perfect field $k$. Our main result shows that this
  algorithm requires at most $O(\max(\deg(\C)^{2\omega}, \deg(D_+)^\omega))$
  arithmetic operations in $k$,
  where $\omega$ is a feasible exponent for matrix multiplication and $D_+$ is the
  smallest effective divisor such that $D_+\geq D$. This improves the
  best known upper bounds on the complexity of computing Riemann-Roch spaces.
  Our algorithm may fail, but we show that provided that a few mild assumptions
  are satisfied, the failure probability is bounded by $O(\max(\deg(\C)^4,
  \deg(D_+)^2)/\lvert \mathcal E\rvert)$, where $\mathcal E$
  is a finite subset of $k$ in which we pick elements uniformly at random.
  We provide a freely available C++/NTL implementation of the proposed algorithm
  and we present experimental data. In particular, our
  implementation enjoys a speedup larger than 6 on many examples (and larger than 200 on some
  instances over large finite
  fields) 
  compared to the reference implementation in the Magma computer algebra system.
  As a by-product, our algorithm also yields a method for computing the group
  law on the Jacobian of a smooth plane curve of genus $g$ within $O(g^\omega)$
  operations in $k$, which equals the best known complexity for this problem.
\end{abstract}

\section{Introduction}

The Riemann-Roch theorem is a fundamental result in algebraic geometry. In its
classical version for smooth projective curves, it provides information on the
dimension of the linear space of functions with some prescribed zeros and
poles. The computation of such Riemann-Roch spaces is a subroutine used in
several areas of computer science and computational mathematics. One of its
most proeminent applications is the construction of algebraico-geometric
error-correcting codes~\cite{goppa1983algebraico}: Such codes are precisely
(subspaces of) Riemann-Roch spaces. Another direct application is the
computation of the group law on the Jacobian of a smooth curve: representing a
point in the Jacobian of a genus-$g$ curve $\C$ as $D - g O$, where $D$ is an
effective divisor of degree $g$ and $O$ is a fixed rational point (or more
generally, a fixed divisor of degree~$1$), the sum of the classes of $D_1 - g
O$ and $D_2-g O$ can be computed by finding a function $f$ in the Riemann-Roch
space $L(D_1+D_2-gO)$. Indeed, by setting $D_3=D_1+D_2-gO+(f)$, the divisor
$D_3 -gO$ is linearly equivalent to $(D_1-gO)+(D_2-gO)$.

\smallskip

{\bf State of the art and related works.} In this paper, we focus on the
classical geometric approach attributed to Brill and Noether for computing
Riemann-Roch spaces. The general algorithmic setting for this approach is
described by Goppa in his landmark paper~\cite[\S 4]{goppa1983algebraico}.
Given a divisor $D$ on a (not necessarily plane) smooth projective curve $\C$, this
method proceeds by finding first a common denominator to all the functions in the
Riemann-Roch space $L(D)$. This is done by computing a form $h$ on the curve
such that the associated principal effective divisor $(h)$ satisfies $(h)\geq
D$. Then the residual divisor $(h)-D$ is computed.  From this, a basis of the
Riemann-Roch space is found by computing the kernel of a linear map. The
correctness of this method is ensured by the residue theorem of Brill and
Noether, which works even in the presence of ordinary singularities by using the
technique of adjoint curves,
see~\cite[\S 42]{severi1921vorlesungen}\cite[Sec.~8.1]{fulton2008algebraic}. In
its original version~\cite[\S 4]{goppa1983algebraico}, Goppa's algorithm works
only for finite fields, and some parts of the algorithm use exhaustive search.
During the 90s, several versions of Goppa's algorithm have been proposed,
incorporating tools of modern computer algebra. In particular, Huang and
Ierardi provide in~\cite{huang1994efficient} a deterministic algorithm for
computing Riemann-Roch spaces of plane curves $\C$ all singularities of which
are ordinary within $O(\deg(\C)^6\deg(D_+)^6)$ arithmetic operations in the
base field, where $D_+$ is the smallest effective divisor such that $D_+\geq
D$. In fact, writing $D_- = D_+-D$, we can assume without loss of generality
that $\deg(D_+)\geq \deg(D_-)$, since $L(D) = L(D_+-D_-)$ is reduced to zero if
$\deg(D)<0$. Consequently, $\deg(D_+)$ is a relevant measure of the size of the
divisor $D$.  Haché~\cite{hache1995computation} proposes the first
implementation of Brill-Noether's approach in a computer algebra system, using
local desingularizations to handle singularities encountered during the
algorithm. For
lines of research closely related to this topic, we refer to~\cite{le1988algorithme,
hache1995effective} and references therein.

A few years later, a breakthrough
is achieved by Hess~\cite{hess2002computing}: He provides an arithmetic
approach to the Riemann-Roch problem, using fast algorithms for algebraic function
fields. Hess' algorithm is now considered as a reference method for
computing Riemann-Roch spaces, and it is proved to be polynomial in the input
size~\cite[Remark~6.2]{hess2002computing}. 

An important special case of the
computation of Riemann-Roch spaces is the computation of the group law on
Jacobians of curves. Volcheck~\cite{volcheck1994computing} describes an algorithm with complexity $O(\max(\deg(\C),g)^7)$ in this context.
The best known complexity for computing the group law on Jacobians of
general curves is currently achieved by Khuri-Makdisi
in~\cite{khuri2007asymptotically}, where he gives an algorithm 
which requires $O(g^{\omega+\varepsilon})$ operations in the base field, where $\omega$
is a feasible exponent for matrix multiplication and $\varepsilon$ is any
fixed positive number. Actually, an anonymous reviewer informed us that the $\varepsilon$
in this complexity can be removed if the cardinality of the base field grows
polynomially in $g$, which is the case in this paper.

\smallskip

{\bf Main results.} We propose a probabilistic algorithm for
computing Riemann-Roch spaces on plane nodal projective curves $\C\subset\P^2$
defined over sufficiently large perfect fields. We emphasize that any algebraic
curve admits such a nodal model up to a birational map if the base field is
sufficiently large (see \emph{e.g.}~\cite[Appendix~A]{arbarello1985geometry}),
and that computing such a model depends only on the curve and not on the input
divisor.

Our main result is that the complexity of the algorithm for computing
Riemann-Roch spaces is bounded by $O(\max(\deg(\C)^{2\omega},
\deg(D_+)^\omega))$ and that, provided that some mild assumptions are
satisfied, its failure probability is bounded above by $O(\max(\deg(\C)^4,
\deg(D_+)^2)/\lvert \mathcal E\rvert)$, where $\mathcal E$ is a finite subset
of the base field $k$ in which we can draw elements uniformly at random.
Roughly speaking, these assumptions on the input require that the impact of the
singularities during the execution of the algorithm is minimal. In particular, they are always
satisfied for smooth curves.  If these mild assumptions are not satisfied, then
the algorithm always fail. Therefore, we provide at the end of
Section~\ref{sec:proba} a Las Vegas procedure (in the sense of
\cite[Sec.~0.1]{babai79}) with complexity
$O(\max(\deg(\C)^5, \deg(D_+)^{5/2}))$ and probability of failure bounded by
$O(\max(\deg(\C)^6, \deg(D_+)^3)/\lvert \mathcal E\rvert)$ to decide whether
these assumptions are satisfied. Combining this verification procedure with our
main algorithm turns it into a complete Las Vegas method, at the cost of
increasing slightly the complexity and the probability of failure.

We also emphasize that our algorithm is geared towards curves defined over
sufficiently large fields $k$, so that the probability of failure can be made
small by choosing a large subset $\mathcal E\subset k$. A possible
workaround to decrease the probability of failure for curves defined over small finite fields
is to do the computations in a field extension, although doing so induces an
extra arithmetic cost.

Up to our knowledge, the complexity that we obtain is the best bound for the general problem of computing
Riemann-Roch spaces. In the special case of the group law on the Jacobian of
plane smooth curves where $\deg(D_+)=O(g)$ and $\deg(\C)=O(\sqrt g)$ by the genus-degree formula, 
the complexity becomes $O(g^\omega)$ which equals the best known complexity
bound of
Khuri-Makdisi's algorithm.
Moreover, the algorithm that we propose requires very few assumptions, and its
efficiency relies on classical building blocks in modern computer algebra: fast
arithmetic of univariate polynomials and fast linear algebra. Consequently, it can be easily made practical by using existing implementations of
these building blocks. We have made a C++/NTL implementation of our algorithm which
is freely distributed under LGPL-2.1+ license and which is available at
\url{https://gitlab.inria.fr/pspaenle/rrspace}.  We also provide experimental
data which seem to indicate that our prototype software is competitive with the reference
implementation in the Magma computer algebra system~\cite{MR1484478}.

\smallskip

{\bf Organization of the paper.} 
Section~\ref{sec:overview} provides an overview of the main algorithm.
Section~\ref{sec:datastruct} focuses on the data structures used to represent
effective divisors. Algorithms to perform additions and subtractions of
divisors with this representation
are described in Section~\ref{sec:div_arith}. Then
Section~\ref{sec:subroutines} gives the details of the subroutines used in the
main algorithm, and their correctness is proved.
Section~\ref{sec:complexity} focuses on the complexity of the subroutines and
of the main algorithm. Then Section~\ref{sec:proba} is devoted to the analysis
of the failure probability. Finally, Section~\ref{sec:expe} presents 
experimental results obtained with our NTL/C++ implementation.

\smallskip

{\bf Acknowledgements.} We are grateful to Simon Abelard, Pierrick Gaudry,
Emmanuel Thomé and Paul Zimmermann for useful discussions and for pointing out important
references.  We thank Pierrick Gaudry for allowing us to use his code for the
fast computation of resultants and subresultants of univariate polynomials. We
are also grateful to an anonymous referee who helped us improve the paper.

\section{Overview of the algorithm}\label{sec:overview}

This section is devoted to the description of the general setting of Brill-Noether's
method and of the algorithm that we propose, without
giving yet all the details on the data structures that we use to represent
mathematical objects.

Throughout this paper, $k$ is a
perfect field and $\C\subset\P^2$ is an absolutely irreducible
projective nodal curve defined over $k$ with $r$ nodes. By nodal curve, we mean that all the
singularities of the curve have order $2$ and are ordinary. We do not need any
assumption about the $k$-rationality of the slopes of the tangents at the nodes. We emphasize that
every algebraic curve admits such a model (up to a field extension if
$k$ is a small finite field), which can be for instance obtained by computing
the image of a nonsingular projective model of the curve by a generic linear
projection to $\mathbb P^2$~\cite[Appendix~A]{arbarello1985geometry}.
We let $\overline k$ denote an
algebraic closure of $k$. Also, we use the notation $\widetilde\C$ to denote a nonsingular model of
$\C$ which projects onto $\C$ (as denoted by $X$
in~\cite[Ch.~8]{fulton2008algebraic}). We assume that this implicit
projection $\widetilde\C\rightarrow \C$ is one-to-one on nonsingular points of
$\C$ and that it is two-to-one on nodes.
By divisor, we always mean a \emph{Weil divisor} on
the curve $\widetilde\C$, i.e. a formal sum with integer coefficients of closed points of
$\widetilde\C$. When the support of a divisor $D$ involves only points of
$\widetilde\C$ which project to nonsingular points of $\C$, we call $D$ a
\emph{smooth divisor} of $\C$ by slight abuse of terminology. More generally, we
will often identify nonsingular closed points of $\C$ with their corresponding points
on $\widetilde C$. We will use frequently the \emph{nodal divisor}, denoted by
$E$, which is the effective
divisor of degree $2r$ which is the sum of all the closed points of
$\widetilde\C$ which project to a node of $\C$.

Naming
$X,Y,Z$ homogeneous coordinates
for $\P^2$, 
the curve $\C\subset\mathbb P^2$ is described by a homogeneous polynomial $Q\in k[X,Y,Z]$
and we let $k[\C]=k[X,Y,Z]/Q(X,Y,Z)$ denote its homogeneous coordinate ring.

Assuming (w.l.o.g. up to linear change of coordinate) that $Q\ne Z$, we let $\C^0\subset\mathbb A^n$ be the affine curve
obtained by intersecting $\C$ with the open subset $\{Z\ne 0\}\subset\P^2$.
It is described
by the bivariate polynomial $q(X, Y)=Q(X,Y,1)$. Closed points of $\C^0$ correspond to maximal
ideals in $k[\C^0] = k[X,Y]/q(X, Y)$. We assume (again w.l.o.g.) that all the nodes of the curve
belongs to its affine subset $\C^0$.

We shall also require that all the divisors that we consider are defined over $k$,
i.e. that they are invariant under the natural action of the Galois group $\Gal(K/k)$
for any extension $K$ of $k$. 
In this setting,
smooth effective divisors on $\C^0$ can be
thought of as nonzero ideals $I$ in
$k[\C^0]$ such that $I+\langle\partial q/\partial X,\partial q/\partial
Y\rangle = k[\C^0]$. For two divisors
$D, D'$ on $\C$, we write $D\leq D'$ if the valuation of $D$ at any place of
$k(\C)$ is at most the valuation of $D'$.
If $g\in k[\C]$ is
a nonzero form
on $\C$, then we let $(g)$ denote the
associated effective principal divisor, as defined in~\cite[Sec.~8.1]{fulton2008algebraic}. 
If $g\in k[\C^0]$ is a nonzero regular function on $\C^0$, then by abuse
of notation, we overload the notation $(g)$ to denote the effective divisor associated to the form
$Z^{\deg(g)} g(X/Z, Y/Z, 1)$. We emphasize that this divisor has no pole and
that it is not the principal divisor associated to the function $g(X/Z, Y/Z,
1)\in k(\C)$.

If $f\in k(\C)$ is a nonzero function on $\C$, i.e. a quotient $f =g/h$
of two nonzero forms
$g,h\in k[\C]$ of the same degree, then again by abuse of notation we let $(f)$ denote the
associated degree-$0$ principal divisor. Finally, for a divisor $D$ we let
$L(D) = \{f\in k(\C)\setminus\{0\}\mid (f)\geq -D\}\cup\{0\}$ denote the Riemann-Roch space
associated to $D$.

\paragraph{Assumptions on the input divisor.} If the curve $\C$ is singular,
then we need two mild assumptions on the input
divisor $D$ to ensure that our algorithm does not always fail. First, the
divisor $D$ should be smooth, and its support should be contained in the affine
chart $\C^0$.
To describe the second assumption --- which is more technical --- we need some
insight on the data structure that we will use: The input divisor $D$ will be
given as a pair of effective divisors $(D_+, D_-)$ such that 
$D = D_+-D_-$. 
Set 
$$d = \begin{cases}
  \lfloor \deg(D_+ + E)/\deg(\C)+(\deg(\C)-1)/2\rfloor\text{ if
  }\binom{\deg(\C)+1}2\leq \deg(D_+ + E)\\
  \lfloor (\sqrt{1+8\deg(D_+ + E)}-1)/2\rfloor\text{\quad\quad\quad\quad~ otherwise.}
\end{cases}$$
We will see in the sequel that this value of $d$ is in fact the smallest
integer which
ensures the existence of a nonzero form $h\in \overline k[\C]$ of degree $d$ such that
$(h)\geq D_+ +E$.
Our second assumption is that there exists a form $h$ of degree $d$ such that
$(h) \geq D_+ +E$ and $(h)-E$ is a smooth divisor. This is mild assumption which is
satisfied in most cases. In the rare cases where it is not satisfied, a
workaround for practical computations --- for which we do not prove any theoretical guarantee of
success --- is to increase slightly
the value of $d$ in Algorithm~\ref{algo:interpolation} (\textsc{Interpolate}) 
in order to increase the dimension of the space of such
functions $h$.

\medskip

\begin{algorithm} 
  \textbf{Function} \textsc{RiemannRochBasis}\;
  \KwData{A curve $\C$ together with its nodal divisor $E$, and a divisor $D=D_+-D_-$ on $\C$ such that $D_+$ and
  $D_-$ are smooth effective divisors.}
  \KwResult{A basis of the Riemann-Roch space $L(D)$.}
  $h\gets$ \textsc{Interpolate}($\deg(\C)$, $D_+$, $E$)\;
  $D_h\gets$ \textsc{CompPrincDiv}($\C$, $h$, $E$)\;
  $D_{\sf res}\gets \textsc{SubtractDivisors}(D_h, D_+)$\;
  $D_{\sf num}\gets \textsc{AddDivisors}(D_-, D_{\sf res})$\;
  $B\gets$ \textsc{NumeratorBasis}($\deg(\C)$, $D_{\sf num}$, $\deg(h)$, $E$)\;
  Return $\{f/h \mid f \in B\}$.
  \caption{A bird's eye view of the algorithm.\label{algo:birdeye}}
\end{algorithm}

Algorithm~\ref{algo:birdeye} gives a bird's eye view of our algorithm for
computing Riemman-Roch spaces. 
We now describe briefly what is done at each step of the algorithm. 
The routine \textsc{Interpolate} takes as input an effective divisor $D_+$, and
it returns a form $h$ such that $(h)\geq D_+ + E$.
Then, \textsc{CompPrincDiv} computes from $h$ a convenient representation of
the divisor $(h)-E$.
The routines used to perform addition and subtraction of divisors --- namely,
\textsc{AddDivisors} and \textsc{SubtractDivisors} --- will be
described in Section~\ref{sec:div_arith}. Then, \textsc{NumeratorBasis} takes as input the
effective divisor $D_{\sf num}$ and the degree of $h$, and it returns a basis
of the vector space of all forms $f\in k[\C]$ of
degree $\deg(h)$ such that $(f)\geq D_{\sf num}+E$.
Finally, we divide this basis by the common denominator $h$ in order to obtain
a basis of the Riemann-Roch space.

One of the cornerstones of the correctness of Algorithm~\ref{algo:birdeye} is
the Brill-Noether's residue theorem. This theorem is one of the foundations of
the theory of adjoint curves. In the case of nodal plane curves, an adjoint
curve is just a curve which goes through all the nodes of $\C$, and $E$ is the
adjoint divisor as defined in~\cite[Sec.~8.1]{fulton2008algebraic}.

\begin{proposition}\cite[Sec.~8.1]{fulton2008algebraic}\label{prop:residue}
Let $D, D'$ be two linearly equivalent effective divisors on $\C$. Let
$h\in k[\C]$ be a form such that $(h)=D+E+A$ for some effective divisor $A$.
Then there exists a form $h'\in k[\C]$ of
the same degree as $h$ such that $(h')=D'+E+A$.
\end{proposition}

We can now prove the general correctness of the main algorithm, assuming that
all the subroutines behave correctly.

\begin{theorem}
  If all the subroutines \textsc{Interpolate}, \textsc{CompPrincDiv},
  \textsc{SubtractDivisors}, \textsc{AddDivisors}, \textsc{NumeratorBasis} are
  correct, then Algorithm~\ref{algo:birdeye} is correct: It returns a basis of the space
  $L(D)$.
\end{theorem}

\begin{proof}
  We first prove that there exists a basis of $L(D)$ such that any basis
  element $f$ belongs to the
  vector space spanned by the output of Algorithm~\ref{algo:birdeye}. To this
  end, we must prove that $f$ can be written as $g/h$ where $h$ is the output
  of the subroutine \textsc{Interpolate} and $g$ belongs to the vector space spanned by
  the output of the subroutine \textsc{NumeratorBasis}. Proposition~\ref{prop:residue} with $D=D_+$, $D'=D_++(f)$ and $h$
  implies that there exists a form $g\in k[\C]$ such that
  $(g/h) = (f)$, where $g$ has the same degree as $h$. Therefore, $f=\lambda g/h$ for some nonzero $\lambda\in k$. It remains to prove
  that $g$ belongs to the vector space spanned by
  the output of \textsc{NumeratorBasis}. Since $f\in L(D)$, we must have
$(g)=(f)+(h)\geq (h)-D = (D_h+E-D_+) + D_-=D_{\rm res}+D_-+E = D_{\rm num}+E$. But
\textsc{NumeratorBasis} returns precisely a basis of the space of forms
$\alpha$ of the
same degree as $h$ such that $(\alpha)\geq D_{\rm num}+E$.

Conversely, let $f$ be a function returned by Algorithm~\ref{algo:birdeye}.
Then $f\cdot h$ belongs to $B$, and hence $(f\cdot h)=(f)+(h)\geq D_{\rm
num}+E=(h)-D$. This implies that $(f)\geq -D$ and hence $f\in L(D)$.
\end{proof}

\section{Data structures}\label{sec:datastruct}

\paragraph{Data structure for the curve $\C$.} We represent the projective
curve $\C\subset\P^2$ by its affine model $\C^0$ in the affine chart $Z\ne 0$
which is described by a bivariate polynomial $q\in k[X,Y]$. 
We assume that the degree of $q$ in $Y$ equals its total
degree. This condition implies that $\C$ is in \emph{projective Noether
position} with respect to the projection on the line $Y=0$, i.e.
that the canonical map $k[X, Z]\rightarrow k[\C]$ is injective and that it
defines an integral ring extension. This also implies that the map
$k[X]\rightarrow k[\C^0]$ is an integral ring
extension. We refer to \cite[Sec.~3.1]{giusti2001grobner} for
more details on the projective Noether position. We emphasize that the
projective Noether
position is achieved in generic coordinates. Hence this assumption does not
lose any generality since it can be enforced
by a harmless linear
change of coordinates. More precisely, regarding a linear change of coordinate
in $\P^2$ as a $3\times 3$ matrix $M=(m_{ij})_{1\leq i,j\leq 3}$, the invertible matrices which put the curve in
projective Noether position are precisely the matrices the $9$ coefficients of
which do not make a polynomial $P(m_{11},\ldots, m_{33})$ of degree $\deg(\C)+3$
vanish. This polynomial is the product of $\det(M)$ --- which has degree $3$
--- with the
coefficient of $Y^{\deg(\C)}$ in the new system of coordinates --- which has degree $\deg(\C)$. Using Schwartz-Zippel
lemma~\cite[Coro.~1]{schwartz1979probabilistic}, this implies that the probability that the 
curve is not in projective Noether position
after a linear change of coordinates given by a
random matrix whose entries are picked uniformly at random in a finite subset
$\mathcal E\subset k$ is bounded above by $(\deg(\C)+3)/\lvert \mathcal
E\rvert$.

\paragraph{Data structure for forms.} We will represent forms on $\C$ --- namely
elements in $k[\C] = k[X, Y, Z]/(Z^{\deg(q)}q(X/Z, Y/Z))$ --- by their affine
counterpart in the affine chart $Z=1$. Consequently, we shall represent a form
$g\in k[\C]$ as an element in $k[X,Y]/q(X, Y)$, given by a representative
$\widetilde g\in k[X,Y]$ such that $\deg_Y(\widetilde g)<\deg(\C)$, using the
fact that $q$ is monic in $Y$.  This representation is not faithful since it
does not encode what happens on the line $Z=0$ at infinity. In order to encode
the behaviour on this line and obtain a faithful representation, it is enough
to adjoin to $\widetilde g$ the degree $d$ of the form $g$, since $g$ is the
class of $Z^d\widetilde g(X/Z, Y/Z)$ in $k[\C]$. In the sequel of this paper,
we do not mention further this issue and we often identify $g$ with $\widetilde
g$ by slight abuse of notation when the context is clear.

\paragraph{Data structure for smooth divisors on $\C$.} For representing
divisors which do not involve any node, 
we use a data structure strongly inspired by the
Mumford representation for divisors on hyperelliptic curves and by
representations of algebraic sets by primitive elements as in~\cite{canny1988some}. Our data structure requires a mild
assumption on the divisor that we represent:
None of the points in
the support of the divisor should lie at infinity. In fact, this is not a
strong restriction since all points can be brought to an affine chart via a
projective change of coordinate. If one does not wish to change the
coordinate system, another solution is to maintain three representations, one for each of the
three canonical affine charts covering $\mathbb P^2$. 

We shall represent a smooth divisor $D$ as a pair of smooth effective divisors $(D_+,
D_-)$ such that $D=D_+-D_-$.
One crucial point for the representation of effective divisors is that the $0$-dimensional algebraic set corresponding to the
support (i.e. without considering the multiplicities) of an effective divisor $D$ can be described by a
finite étale algebra which is a quotient of $k[\C^0]$ by a nonzero ideal. This étale algebra is isomorphic to
a quotient of a univariate polynomial ring if it admits a primitive element. Using primitive elements to represent $0$-dimensional
algebraic sets is a classical technique in computer algebra, see
e.g.~\cite[Sec.~2]{canny1988some}\cite[Sec.~3.2]{giusti2001grobner}.

\begin{lemma}\label{lem:primeltdef}
  Let $R$ be a finite étale $k$-algebra, i.e. a finite product of finite
  extensions of $k$. Let $z\in R$ be an
  element, and let $m_z$ denote the multiplication by $z$ in $R$, seen as a
  $k$-linear endomorphism.
  The following statements are equivalent:
  \begin{enumerate}[itemsep=1pt, topsep=4pt]
    \item The element $z$ generates $R$ as a $k$-algebra;
    \item The elements $1, z, z^2, \ldots, z^{\dim_k(R)-1}$ are linearly
      independent over $k$;
    \item The characteristic polynomial of $m_z$ equals its minimal polynomial;
    \item The characteristic polynomial of $m_z$ is squarefree.
  \end{enumerate}

  If $z$ satisfies these four properties, then $z$ is called a \emph{primitive
  element} for $R$. 
\end{lemma}

\begin{proof}
$(2)\Rightarrow (1)$: By definition, the element $z$ generates $R$ as a $k$-algebra if and
only if its powers generates $R$ as a $k$-vector space. 
$(1)\Rightarrow(2)$: Let $n_0$ be the smallest positive integer such that $1, z,
z^2, \ldots, z^{n_0}$ are linearly dependent. The integer $n_0$ must be finite since
$\dim_k(R)$ is finite. Write $z^{n_0} = \sum_{i=0}^{n_0-1} a_i z^i$ for some
$a_0,\ldots, a_{n_0}\in k$. By multiplying this relation by $z^{n-n_0}$ and by
induction on $n$, we obtain that for any $n\geq n_0$, $z^n$ belongs to the
vector space generated by $1, z, \ldots, z^{n_0-1}$. This implies that the
algebra generated by $z$ has dimension $n_0$ as a $k$-vector space. By $(1)$,
we obtain that $n_0=\dim_k(R)$.
$(2)\Rightarrow(3)$: By $(2)$, the minimal polynomial of $m_z$ has degree at
least $\dim_k(R)$, and hence it equals its characteristic polynomial.
$(3)\Rightarrow(2)$: The degree of the characteristic polynomial is $\dim_k(R)$.
$(3)\Rightarrow(4)$: Let $\xi$ be the squarefree part of the characteristic
polynomial of $m_z$. By $(3)$, $\xi(z)$ must be nilpotent in $R$. But the only
nilpotent element in an étale algebra in $0$, so $\xi$ must be a multiple of
the minimal polynomial of $m_z$. Hence, by $(3)$, $\xi$ is the characteristic polynomial
of $m_z$. $(4)\Rightarrow(3)$: This is a consequence of the facts that the characteristic polynomial and the
minimal polynomial have the same set of roots, and the minimal polynomial
divides the characteristic polynomial. \end{proof}

We are now ready to define the data structure that we will use to represent
smooth effective divisors on the curve.
A smooth effective divisor $D$ on $\C$ supported on the affine chart $Z\ne 0$ will be represented as:
\begin{itemize}[noitemsep, topsep=2pt]
  \item A scalar $\lambda\in k$;
  \item Three univariate polynomials $\chi, u, v\in k[S]$, such that $\chi$ is
    monic, $
    \chi$ has degree $\deg(D)$ and $u, v$ have degree at most $\deg(D)-1$.
\end{itemize}

such that

\begin{description}[labelwidth=1.5cm, noitemsep, topsep=2pt]
    \itemsep0em
  \item[(Div-H1)] $q(u(S), v(S))\equiv 0\bmod \chi(S)$;
  \item[(Div-H2)] $\lambda u(S)+ v(S) = S$;
  \item[(Div-H3)] $\GCD(\frac{\partial q}{\partial X}(u(S), v(S))-\lambda \frac{\partial q}{\partial Y}(u(S),
    v(S)),\chi(S)) = 1$.
\end{description}

We call
the data structure above a \emph{primitive element representation}.
An
important ingredient of the primitive element representation is that \textbf{(Div-H3)} enables us to use Hensel's lemma to
encode the multiplicites. More precisely, \textbf{(Div-H3)} implies that at
each of the closed points in the support of the divisor, the element
$\lambda (X-\bar x)+(Y-\bar y)$ is a uniformizing element for the associated discrete
valuation ring, where
$\bar x, \bar y$ denote the classes of $X, Y$ in the residue
field.

Notice that this representation requires the existence of a primitive element of the form $\lambda X+Y$ which satisfies all the wanted properties.
Fortunately, Proposition~\ref{prop:existence_primrepr} below shows that such a primitive element
exists as soon as $k$ contains more than $\binom{\deg(D)+1}2$ elements.

\paragraph{Data structure for the nodal divisor.} We shall represent the nodal
divisor via an algebraic parametrization by the roots of a
univariate polynomial. This algebraic structure is very similar to the
representation of smooth divisors, except for a crucial difference: We shall
not need to represent multiplicities, so this representation does not need to satisfy condition
\textbf{(Div-H3)}. More precisely, the nodal divisor $E$ will be represented as:
\begin{itemize}[noitemsep, topsep=2pt]
  \item A scalar $\lambda_E\in k$;
  \item Three univariate polynomials $\chi_E, u_E, v_E\in k[S]$, such that
    $\chi_E$ is
    monic and squarefree, $\chi_E$ has degree $r$ and $u, v$ have degree at
    most $r-1$;
  \item A monic univariate polynomial $T_E\in k[S]$ of degree at most $2r$\end{itemize}

such that
\begin{equation}\tag{\textbf{NodDiv-H1}}\{(u_E(\zeta), v_E(\zeta))\mid \zeta\in
  \overline k,\chi_E(\zeta)=0\}\subset
  {\overline k}^2\text{ is the set of nodes of $\C^0$,}
\end{equation}
and such that the roots
    of $T_E$ are the values $\lambda\in \overline k$ such that the vector
    $(1, -\lambda)$ is tangent to $\C^0$ at a node. Notice that the roots of
    $T_E$ do not record vertical tangents at nodes, so the degree of $T_E$ may be
    less than $2r$.

Such $(\lambda_E, \chi_E, u_E, v_E)$ satisfying \textbf{(NodDiv-H1)} exist as soon as $k$ contains more
than $\binom{r}{2}$ elements by
Proposition~\ref{prop:existence_primrepr_nodal}. 
Computing $T_E$ is also an easy task once $(\lambda_E, \chi_E, u_E, v_E)$ are
known. This polynomial can be for instance obtained by considering the
homogeneous form $Q_2(X,Y,S)$ of degree $2$ of the
shifted polynomial $q(X+u_E(S), Y+v_E(S))$. Then
$T_E(\lambda) = \Resultant_S(Q_2(1, -\lambda, S), \chi_E(S))$ satisfies the
desired property. This polynomial $T_E$ will be useful in
Algorithm~\textsc{CompPrincDiv}.
Computing $\lambda_E, \chi_E, u_E, v_E, T_E$ can be thought of as a precomputation since
this depends only on the curve $\C$ and not on the input divisor $D$.

Before going any further, we summarize here the data structures for the input
and the output of
Algorithm~\ref{algo:birdeye} and the properties that they must satisfy.

\medskip

Input data:
\begin{itemize}[noitemsep, topsep=4pt]
  \item A bivariate polynomial $q\in k[X,Y]$. This
    polynomial encodes the curve $\C$;
  \item Data $(\lambda_E, \chi_E, u_E, v_E, T_E)$ encoding the nodal divisor $E$;
  \item A smooth divisor $D=D_+-D_-$ given by two tuples $(\lambda_+, \chi_+, u_+, v_+)$ and 
    $(\lambda_-, \chi_-, u_-, v_-)$ with $\lambda_\pm\in k$, 
    $\chi_\pm, u_\pm, v_\pm\in k[S]$.
\end{itemize}

The input data must satisfy the following constraints:
\begin{enumerate}[noitemsep, topsep=4pt]
  \item The bivariate polynomial $q\in k[X,Y]$ is absolutely irreducible, and
    its base field $k$ is perfect;
  \item The total degree of $q$ equals its degree with respect to $Y$;
\item The inequalities $\deg(u_\pm)<\deg(\chi_\pm)$,
  $\deg(v_\pm)<\deg(\chi_\pm)$, $\deg(u_E)<\deg(\chi_E)$,
  $\deg(v_E)<\deg(\chi_E)$ hold;
\item The polynomials $\chi_\pm$, $\chi_E$ and $T_E$ are monic;
\item The polynomial $\chi_E$ is squarefree;
\item Both tuples $(\lambda_+, \chi_+, u_+, v_+)$ and
  $(\lambda_-, \chi_-, u_-, v_-)$ satisfy \textbf{(Div-H1)} to
  \textbf{(Div-H3)};
\item The tuple $(\lambda_E, \chi_E, u_E, v_E, T_E)$ satisfies \textbf{(NodDiv-H1)};
\item The degree of $T_E$ is at most $2r$;
\item The roots of the univariate polynomial $T_E$ are the values
  $\lambda\in\overline k$ such that the vector $(1, -\lambda)$ is tangent to
  $\C^0$ at a node.
\end{enumerate}

\medskip

Output data:
\begin{itemize}[noitemsep, topsep=4pt]
  \item A bivariate polynomial $h\in k[X,Y]$;
  \item A finite set of bivariate polynomials $B\subset k[X,Y]$.
\end{itemize}

The output data satisfies that the set $\{b/h \mid b \in B\}$ is a basis of the
Riemann-Roch space associated to $D$ on $\C$.

The rest of this section is devoted to technical proofs about the
primitive element representation. The statements below will be used for proving
the correctness of the subalgorithms, but they may be skipped without harming the
general understanding of this paper.

\smallskip

The two following propositions (whose proofs are postponed after
Lemma~\ref{lem:existence_primrepr_aux}) show that primitive representations of smooth effective
divisors and of the nodal divisor exist
provided that the base field is large enough.

\begin{proposition}\label{prop:existence_primrepr}
  Let $J$ be a nonzero ideal of $k[\C^0]=k[X,Y]/q(X,Y)$ such that
  $J+\langle\frac{\partial q}{\partial X}, \frac{\partial q}{\partial
  Y}\rangle = k[\C^0]$. 
  Assume that the cardinality of $k$ is larger than $\binom{\dim_k(k[\C^0]/J)+1}2$.
  Then there exist $\lambda\in k$ and polynomials $\chi, u,
  v\in k[S]$ satisfying \textbf{(Div-H1)} to \textbf{(Div-H3)} such that the
  map
    $k[\C^0]/J\rightarrow k[S]/\chi(S)$ sending $X$ and $Y$ to the classes of
    $u$ and $v$ is an isomorphism of $k$-algebras.
\end{proposition}

The following proposition states a similar result for radical ideals which
do not satisfy the smoothness assumption. This will be useful to represent the
nodal divisor.

\begin{proposition}\label{prop:existence_primrepr_nodal}
  Let $J$ be a nonzero radical ideal of $k[\C^0]=k[X,Y]/q(X,Y)$. 
  Assume that the cardinality of $k$ is larger than $\binom{\dim_k(k[\C^0]/J)}2$.
  Then there exist $\lambda_E\in k$ and polynomials $\chi_E, u_E,
  v_E\in k[S]$ satisfying \textbf{(NodDiv-H1)} such that the
  map
    $k[\C^0]/J\rightarrow k[S]/\chi(S)$ sending $X$ and $Y$ to the classes of
    $u$ and $v$ is an isomorphism of $k$-algebras.
\end{proposition}

Before proving Propositions~\ref{prop:existence_primrepr}
and~\ref{prop:existence_primrepr_nodal}, 
we need some technical lemmas.
First, the following lemma generalizes
slightly
the classical fact that ideals in the coordinate rings of smooth curves admit a unique factorization.
Here, we do not assume that $\C^0$ is nonsingular, but the
factorization property holds only for ideals of regular functions which do not vanish
at any singular point.

\begin{lemma}\label{lem:almostdedekind}
  Let $I$ be a nonzero ideal of $k[\C^0]$ such that $I + \langle
  \partial q/\partial X, \partial
  q/\partial Y\rangle = k[\C^0]$. Then there exists a unique
  factorization $I = \prod_{i=1}^\ell \mathfrak m_i^{\alpha_i}$ as a product of maximal
  ideals of $k[\C^0]$.
\end{lemma}

\begin{proof}
First, we
prove the existence of such a factorization. Let $\mathfrak m\subset k[\C^0]$ be an ideal containing $I$. If $\mathfrak m$
is a nonsingular closed point of $\C^0$, then the local ring
$k[\C^0]_{\mathfrak m}$ is a discrete valuation ring by~\cite[Sec.~3.2,
Thm.~1]{fulton2008algebraic}.
Let $\val_{\mathfrak m}(I)\in\mathbb Z_{\geq 0}$ denote the integer such that
$I=\mathfrak m^{\val_{\mathfrak m}(I)}$ in this local ring.
Let $J$ be the ideal
$J = \displaystyle\prod_{\mathfrak m\supset I}\mathfrak m^{\val_{\mathfrak
m}(I)}.$
By~\cite[Prop.~9.1]{atiyah1969introduction}, the equality
$I = J$
holds if and only if it holds in all the local rings
$k[\C^0]_{\mathfrak m}$ where $\mathfrak m\supset I$. Since the maximal ideals
$\mathfrak m$
are nonsingular closed points of $\C^0$, this equality holds true because
the corresponding local rings are discrete valuation rings, hence the equality of ideals is equivalent to
the equality of their $\mathfrak m$-valuation.

We now prove the unicity of this factorization: By contradiction, assume that
$I$ has two distinct factorizations
$\prod_{1\leq i\leq \ell} \mathfrak m_i^{\alpha_i}$ and
$\prod_{1\leq i\leq \ell'}
\mathfrak m_i'^{\alpha_i'}$ of the ideal. Without loss of generality, assume that $\mathfrak m_1$ does not occur in the second factorization, or that
it appears with a different multiplicity. This would lead to a contradiction
since it would lead to distinct valuations of the same ideal in the local ring
at $\mathfrak m_1$. \end{proof}

An ideal $I\subset k[X,Y]$ in a polynomial ring is
called $0$-dimensional if the dimension of $k[X,Y]/I$ as a $k$-vector
space is finite. The following lemma identifies values of $\lambda$
for which $\lambda X+Y$ is not a primitive element for a $0$-dimensional
algebraic set.

\begin{lemma}\label{lem:primeltbidegbound}
Let $I\subset k[X, Y]$ be a radical
$0$-dimensional ideal, with associated variety $V = \{\alpha_i\}_{1\leq i\leq
\dim_k(k[X,Y]/I)}\subset\overline k^2$ and let
$u,v\in k[X,Y]$ be elements such that for any distinct points
$\alpha_i,\alpha_j\in V$, $u(\alpha_i)\ne u(\alpha_j)$ or
$v(\alpha_i)\ne v(\alpha_j)$. Then the set of $\lambda\in k$ such that
$\lambda u + v$ is not a primitive element for $k[X,Y]/I$ is contained in the set of roots of
a nonzero univariate polynomial with coefficients in $k$ of degree $\binom{\dim_k(k[X,Y]/I)}2$.
\end{lemma}

\begin{proof}
  Writing $\alpha_i=(\alpha_{i,x},\alpha_{i,y})$, let $K\subset \overline k$ be a field extension of $k$ where $I$ factors as a product of
  degree-$1$ maximal ideals $\mathfrak m_i=\langle X-\alpha_{i,x},
   Y-\alpha_{i,y}\rangle$. This provides an isomorphism of
  $K$-algebras between $K[X,Y]/I$ and $K^{\dim_k (k[X_1,X_2]/I)}$ sending polynomials to their evaluations
  at $\alpha_1,\ldots, \alpha_{\dim_k(k[X,Y]/I)}$.
Using this isomorphism, and letting $\mathbf e_i$ denote the $i$-th canonical vector in $K^{\dim_k (k[X_1,X_2]/I)}$,
we observe that $\mathbf e_i$ is an eigenvector of the endomorphism of
multiplication by $\lambda u+v$, with associated eigenvalue $\lambda
u(\alpha_i)+v(\alpha_i)$. Next, $\lambda u+v$ is a primitive element for
$k[X,Y]/I$ if all these eigenvalues are distinct by
Lemma~\ref{lem:primeltdef}. This is the case if and only if the discriminant of
the characteristic polynomial is nonzero. Since the discriminant is the
product of the squared differences of the roots, it equals
$$\prod_{\substack{1\leq i\leq \dim_k(k[X,Y]/I)\\1\leq j<i}} (\lambda
(u(\alpha_i)-u(\alpha_j)) + v(\alpha_i) - v(\alpha_j))^2.$$
This discriminant is a polynomial in $k[\lambda]$ of degree
$2\binom{\dim_k(k[X,Y]/I)}2$. It is nonzero because of the assumption that for any
distinct $\alpha_i, \alpha_j$, either $u(\alpha_i)\ne u(\alpha_j)$ or
$v(\alpha_i)\ne v(\alpha_j)$. Let $\Delta(\lambda)=\prod_{1\leq j<i} (\lambda
(u(\alpha_i)-u(\alpha_j)) + v(\alpha_i) - v(\alpha_j))\in\overline k[\lambda]$ be its squareroot.
It remains to prove that the polynomial $\Delta$ has coefficients in $k$. This
is due to the fact that automorphisms in $\Gal(K/k)$
permute the points $\alpha_i$. Therefore the natural action of $\Gal(K/k)$
acts on $\Delta$ by permuting its linear factors, and hence it leaves $\Delta$ invariant.
\end{proof}

In the following lemma, the notation $\red(R)$ stands for the quotient of a ring
$R$ by its Jacobson radical (i.e. the intersection of its maximal ideals).
If $I$ is an ideal of $R$, we use the notation $\sqrt I$ to denote the
radical of $I$. The
ring $\red(k[\C^0]/J)$ can be thought of as the coordinate ring
of the $0$-dimensional algebraic set corresponding to the points in the support
of the effective divisor associated to $J$.

\begin{lemma}\label{lem:existence_primrepr_aux}
  Let $J$ be a nonzero ideal of $k[\C^0]=k[X,Y]/q(X,Y)$ such that
  $J+\langle \partial q/\partial X, \partial q/\partial
  Y\rangle = k[\C^0]$. 
  Then there exists a nonzero univariate polynomial $\Delta$ with coefficients
  in $k$ of degree at most $\binom{\dim_k(k[\C^0]/J)+1}2$, such that for any
  $\lambda$ which is not a root of $\Delta$, the element $\lambda
  X+Y$ is primitive for $\red(k[\C^0]/J)$ and
  $\partial q/\partial X-\lambda\partial q/\partial Y$ is
  invertible in $\red(k[\C^0]/J)$. 
  
  In particular, if $k$ has cardinality larger
  than $\binom{\dim_k(k[\C^0]/J)+1}2$, then there exists a value of $\lambda$
  in $k$ which is not a root of $\Delta$.
\end{lemma}

\begin{proof}
  First, notice that $k[\C^0]/J$ is isomorphic to $k[X,Y]/(J+\langle
  q\rangle)$, by using the classical fact that ideals of a quotient ring $R/I$
  correspond to ideals of $R$ containing $I$. Since $q$ is irreducible, $J+\langle
  q\rangle\subset k[X,Y]$ is a zero-dimensional ideal, and hence $\red(k[X,Y]/(J+\langle
  q\rangle)) = k[X,Y]/\sqrt{J+\langle
  q\rangle}$. Notice that $\dim_k(k[X,Y]/\sqrt{J+\langle
  q\rangle}) \leq \dim_k(k[\C^0]/J)$. Next, using the fact that two distinct
  points in the variety have distinct coordinates, Lemma~\ref{lem:primeltbidegbound}
  provides a nonzero polynomial $\Delta_0$ of degree at most
  $\binom{\dim_k(k[\C^0]/J)}2$ such that $\lambda X+Y$ is not a primitive
  element for $\red(k[\C^0]/J)$ only if $\lambda$ is a root of $\Delta_0$.

  Next, we notice that since $\sqrt{J+\langle
  q\rangle}\subset k[X,Y]$ is a radical $0$-dimensional ideal, it can be decomposed as a product $\prod_{1\leq
  i\leq \ell}\mathfrak m_i$ of maximal ideals. Consequently, $\partial q/\partial X-\lambda\partial
  q/\partial Y$ is invertible in
  $\red(k[\C^0]/J)$ if and only if $\partial q/\partial X-\lambda\partial
  q/\partial Y$ does not belong to any of these
  maximal ideals. Equivalently, $\partial q/\partial X-\lambda\partial
  q/\partial Y$ must not vanish in any of the
  residue fields $\kappa_i=k[\C^0]/\mathfrak m_i$. Notice that the norm
  $N_{\kappa_i/k}(\partial q/\partial X-\lambda\partial
  q/\partial Y)$ is a polynomial $\Delta_i$ in
  $\lambda$ with coefficients in $k$. It is nonzero since
  $J+\langle \partial q/\partial X, \partial q/\partial
  Y\rangle = k[\C^0]$ and hence either $\partial q/\partial X$ or
  $\partial q/\partial Y$ is nonzero in $\kappa_i$. Therefore $\Delta_i$ is
  either constant (if $\partial
  q/\partial Y$ vanishes in $\kappa_i$), or it has degree $[\kappa_i:k]$.
  Finally, the proof is concluded by noticing that $\sum_{1\leq i\leq \ell}
  [\kappa_i:k] = \dim_k(k[X,Y]/\sqrt{J+\langle
  q\rangle}) \leq \dim_k(k[\C^0]/J)$, so that the product
  $\Delta_0\cdot\prod_{1\leq i\leq\ell}\Delta_i$ has degree at most
  $\binom{\dim_k(k[\C^0]/J)+1}2$ and satisfies all the desired properties.
\end{proof}

We now have all the tools that we need to prove
Propositions~\ref{prop:existence_primrepr}
and~\ref{prop:existence_primrepr_nodal}.

\begin{proof}[Proof of Proposition~\ref{prop:existence_primrepr}]
  First, we assume that $J=\mathfrak m^\alpha$ is a power of a maximal ideal in
  $k[\C^0]$ such that $\mathfrak m+\langle\partial q/\partial X,\partial
  q/\partial Y\rangle = k[\C^0]$. Then $\red(k[\C^0]/J) = k[\C^0]/\mathfrak m$. Let $\lambda\in k$
  be an element which is not a root of the polynomial $\Delta$ provided by
  Lemma~\ref{lem:existence_primrepr_aux}. Such an element exists
  since the cardinality of $k$ is larger than the degree of $\Delta$. Therefore, $\lambda
  X+Y$ is a primitive element for $k[\C^0]/\mathfrak m$ and hence there exist
  univariate polynomials $\widetilde u,\widetilde v\in
  k[S]$ 
  such that $X=\widetilde u(\lambda X+Y)$ and $Y=\widetilde v(\lambda X+Y)$ in $k[\C^0]/\mathfrak
  m$. Let $\widetilde \chi(S)$ be the minimal polynomial of $\lambda X+Y$ in
  $k[\C^0]/\mathfrak m$, which is irreducible since $k[\C^0]/\mathfrak m$ is a
  field. 
  Notice that the map $k[\C^0]/\mathfrak m \rightarrow k[S]/\widetilde \chi(S)$
  sending the classes of $X, Y$ to $\widetilde u, \widetilde v$ is an isomorphism of
  $k$-algebras.
  Next, set $\chi(S) =
  \widetilde\chi(S)^\alpha$ and consider the bivariate system
  \begin{equation}\label{eq:henselsys}
  \begin{cases}q(X,Y) = 0\\
  \lambda X+Y-S = 0.
\end{cases}
\end{equation}
By construction, this system has solution $(\widetilde u,\widetilde v)$ over $k[S]/\widetilde\chi(S)$. The
Jacobian of this system is $\frac{\partial
q}{\partial X}(X,Y)-\lambda\frac{\partial q}{\partial Y}(X,Y)$, which is invertible in $\red(k[\C^0]/J)$ by Lemma~\ref{lem:existence_primrepr_aux}, and therefore
$\frac{\partial
q}{\partial X}(\widetilde u(S),\widetilde v(S))-\lambda\frac{\partial
q}{\partial Y}(\widetilde u(S),\widetilde v(S))$ is invertible in $k[S]/\widetilde\chi(S)$. By Hensel's
lemma, 
there exist polynomials $u,v\in k[S]_{<\deg(\chi)}$ which are solutions of~\eqref{eq:henselsys} over $k[S]/\chi(S)$: 
Indeed, for $i>1$, if $(\widehat u, \widehat v)$ is a solution of~\eqref{eq:henselsys} over
$k[S]/\widetilde\chi(S)^i$, then a Taylor expansion of the system at order $1$ shows that
$$\begin{bmatrix}\widehat u\\
  \widehat v\end{bmatrix}-\begin{bmatrix}
	\partial q/\partial X& \partial q/\partial Y\\
	\lambda& 1
\end{bmatrix}^{-1}\cdot \begin{bmatrix}
  q(\widehat u,\widehat v)\\
  \lambda \widehat u+\widehat v-S
 \end{bmatrix}$$
 is a solution of Eq.~\eqref{eq:henselsys} over 
 $k[S]/\widetilde\chi(S)^{2i}$.

The map $k[\C^0]/J\rightarrow
k[S]/\chi(S)$ is well-defined because $\mathfrak m$ maps to $0$ modulo
$\widetilde\chi$ and hence $J=\mathfrak m^\alpha$ maps to $0$ modulo
$\chi=\widetilde\chi^\alpha$. It is an isomorphism because $k[\C^0]/J$ and
$k[S]/\chi(S)$ have the same dimension as vector spaces over $k$ and the map
$S\mapsto \lambda X+Y$ is the right inverse to the map
$(X,Y)\mapsto(u(S),v(S))$. It remains to prove that \textbf{(Div-H3)} is
  satisfied by $\chi, u, v$, which is a direct consequence of the fact that by
  Lemma~\ref{lem:existence_primrepr_aux}, $\partial q/\partial
  X-\lambda \partial q/\partial Y$ does not belong to $\mathfrak m$ and hence 
  $\frac{\partial q}{\partial X}(u(S), v(S))-\lambda \frac{\partial q}{\partial
  Y}(u(S),
    v(S))$ is invertible modulo $\chi(S)$.

  Next, we consider the general case where $J$ is a nonzero ideal in
  $k[\C^0]$ such that $J+\langle\partial q/\partial X,\partial
  q/\partial Y\rangle = k[\C^0]$. Again, let $\lambda\in k$
  be an element which is not a root of the polynomial $\Delta$ provided by
  Lemma~\ref{lem:existence_primrepr_aux}. Lemma~\ref{lem:almostdedekind} implies that $J$ can be written as a
  product $J=\prod_{i=1}^\ell \mathfrak m_i^{\alpha_i}$ of powers of maximal
  ideals. Then for all $i$, the element $\lambda X + Y$ is
primitive for $k[\C^0]/\mathfrak m$ and $\partial q/\partial X - \lambda \partial
q/\partial Y$ is invertible in $k[\C^0]/\mathfrak m$. For each $i$, using the previous argument,
  we can construct univariate polynomials $\chi_i, u_i, v_i\in k[S]$ satisfying
  \textbf{(Div-H1)} to \textbf{(Div-H3)} with respect to $\lambda$ such that
  the maps $k[\C^0]/\mathfrak m_i^{\alpha_i} \rightarrow k[S]/\chi_i(S)$
  sending $X, Y$ to $u_i(S), v_i(S)$ are isomorphisms of $k$-algebras. Setting
  $\chi(S) = \prod_{i=1}^\ell \chi_i(S)$ and using the CRT, let $u,v \in
  k[S]<\deg(\chi)$
  be such that for all $i$, we have $u(S)\equiv u_i(S)\bmod \chi_i(S)$ and
  $v(S)\equiv v_i(S)\bmod \chi_i(S)$. Then the fact that the CRT is a ring
  morphism
  allows us to conclude that the map $k[\C^0]/J \rightarrow k[S]/\chi(S)$ is an
  isomorphism and that $\chi, u, v$ satisfy \textbf{(Div-H1)} to
  \textbf{(Div-H3)}.
\end{proof}

\begin{proof}[Proof of Proposition~\ref{prop:existence_primrepr_nodal}]
The proof is similar to that of Proposition~\ref{prop:existence_primrepr},
by ignoring the argument about multiplicities. 
Since the
Jacobian $\frac{\partial q}{\partial X}(X,Y)-\lambda\frac{\partial q}{\partial
Y}(X,Y)$ need not be
invertible in $\red(k[\C^0]/J)$, it is sufficient to choose a value of
$\lambda$ which is not a root of the univariate polynomial constructed in 
Lemma~\ref{lem:primeltbidegbound} for $J$.
\end{proof}

The next lemma shows that any data satisfying \textbf{(Div-H1)} to
\textbf{(Div-H3)} actually encodes a well-defined effective divisor with no
singular point in its support.

\begin{lemma}\label{lem:divHprimelt} Let $(\lambda, \chi, u, v)$ be such that
  \textbf{(Div-H1)} to \textbf{(Div-H3)} are satisfied, and let $I=\langle
  X-u(S), Y-v(S),
  \chi(S)\rangle\cap k[X,Y]$.  Then $k[X,Y]/I$ is isomorphic as a $k$-algebra
  to $k[\C^0]/J$ where $J$ is a nonzero ideal in $k[\C^0]$.
  Moreover, $J+\langle \partial q/\partial X,\partial q/\partial
  Y\rangle=k[\C^0]$, $\lambda X+Y$ is a primitive element for
  $\red(k[X,Y]/I)$, and its minimal polynomial is the squarefree part of
  $\chi$.
\end{lemma}

\begin{proof}
  Nonzero ideals of $k[\C^0]$ correspond to ideals of $k[X,Y]$ containing
  properly the
  principal ideal $\langle q(X,Y)\rangle$. First, notice that \textbf{(Div-H1)}
  implies that $q(X,Y)\in I$. Also, by \textbf{(Div-H2)}, we get that
  $\chi(\lambda X+Y)\in I$. Notice that $\chi(\lambda X+Y)$ factors as a product of
  polynomials of degree $1$ over the algebraic closure of $k$. Since $q$ is supposed to be
  absolutely irreducible and to have degree at least $2$, this implies that
  $\chi(\lambda X+Y)$ does not belong to the principal ideal $\langle
  q(X,Y)\rangle$. Consequently, $I$ contains properly $\langle q(X,Y)\rangle$
  and this proves the isomorphism between $k[X,Y]/I$ and $k[\C^0]/J$. In
  particular, we obtain that $\dim_k(k[\C^0]/J) = \dim_k(k[X,Y]/I)=\deg(\chi)$.
  Next, \textbf{(Div-H3)} implies that $\partial q/\partial X -
  \lambda\partial q/\partial
  Y$ is invertible in $k[\C^0]/J$, and hence $k[\C^0] = J+\langle\partial q/\partial X -
  \lambda\partial q/\partial
  Y\rangle\subset J+\langle \partial q/\partial X,\partial q/\partial
  Y\rangle$. Therefore, $J+\langle \partial q/\partial X,\partial q/\partial
  Y\rangle=k[\C^0]$. Using the
  isomorphism between $k[\C^0]/J$ and $k[S]/\chi(S)$ described in
  Proposition~\ref{prop:existence_primrepr}, we obtain that $\red(k[X,Y]/I)$ is
  isomorphic to $\red(k[S]/\chi(S))$, which is in turn isomorphic to
  $k[S]/\widetilde\chi(S)$, where $\widetilde\chi(S)$ is the squarefree part of
  $\chi(S)$. Finally, the proof is concluded by noticing that $S$ is a primitive
  element for $k[S]/\widetilde\chi(S)$ with minimal polynomial
  $\widetilde\chi(S)$.
\end{proof}

The following lemma explicits the link between the primitive element
representation and the ideal vanishing on the $0$-dimensional algebraic set
that it represents.

\begin{lemma}\label{lem:techrepr}
  Let $(\lambda, \chi, u, v)$ be data satisfying \textbf{(Div-H2)}. Set $I=\langle \chi(S), X-u(S), Y-v(S)\rangle\subset k[X, Y,
S]$ and $J = \langle \chi(\lambda
X+Y), X-u(\lambda X+Y), Y-v(\lambda X+Y)\rangle$. Then
$I\cap k[X,Y] = J$.
\end{lemma}
\begin{proof}
  By \textbf{(Div-H2)} and by using the fact that $X-u(S), Y-v(S)\in I$, we
  deduce that $S-(\lambda X+Y)\in I$. This implies that $I\cap
  k[X,Y]=\{f(X,Y,\lambda X+Y)\mid f\in I\}$.
\end{proof}

The primitive element representation of an effective divisor is not unique: Two tuples $(\lambda_1, \chi_1, u_1, v_1)$ and 
$(\lambda_2, \chi_2, u_2, v_2)$ may encode the same effective
divisor. The cases where this happens are detailed in the following
proposition.

\begin{proposition}\label{prop:primelteq}
Let $(\lambda_1,
\chi_1, u_1, v_1), (\lambda_2,
\chi_2, u_2, v_2)\in k\times k[S]^3$ be data which satisfy \textbf{(Div-H2)}.
Let $I_1, I_2\subset k[X, Y, S]$ be the associated ideals $I_1 =
\langle \chi_1(S), X - u_1(S), Y-v_1(S)\rangle, I_2 =
\langle \chi_2(S), X - u_2(S), Y-v_2(S)\rangle$.

Then $I_1\cap k[X,Y] = I_2\cap k[X,Y]$ if and only if $\chi_1$ is
the characteristic polynomial of $\lambda_1u_2+
v_2$ in $k[S]/\chi_2(S)$, $u_1(\lambda_1 u_2(S)+
v_2(S))\equiv u_2(S)\bmod \chi_2(S)$ and $v_1(\lambda_1 u_2(S)+
v_2(S))\equiv v_2(S)\bmod \chi_2(S)$.
\end{proposition}

\begin{proof}We first prove
  the ``if'' part of the statement. 
  First, we notice that $k[X, Y]/(I_1\cap k[X,Y])$ and
  $k[X,Y]/(I_2\cap k[X,Y])$ are $k$-vector space of the same finite dimension,
  since $\deg(\chi_1)$ must equal $\deg(\chi_2)$.
  Therefore it is enough to show one inclusion to prove the equality. Let $f(X, Y)\in I_1\cap k[X,Y]$.
  Using the
  equalities modulo $\chi_2$, we obtain that $f(X, Y)\equiv f(u_1(\lambda_1 u_2(S)+
v_2(S)), v_1(\lambda_1 u_2(S)+
v_2(S)))\bmod I_2$, which is divisible by $\chi_1(\lambda_1 u_2(S)+
v_2(S))$ because $f$ is in $I_1$ and by using Cayley-Hamilton theorem. Finally, we use the fact that
$\chi_1(S)$ is the characteristic polynomial of $\lambda_1 u_2(S)+
v_2(S)$ and hence $\chi_2$ divides $\chi_1(\lambda_1 u_2(S)+
v_2(S))$, which finishes to prove that $f\in I_2$.

Conversely, assume
that $I_1\cap k[X,Y] = I_2\cap k[X,Y]$.
By composing the isomorphisms
$$\begin{array}[t]{ccc}
  k[S]/\chi_1(S)&\rightarrow& k[X,Y]/(I_1\cap k[X,Y])\\
  S&\mapsto&\lambda_1 X+Y
\end{array}\quad
\begin{array}[t]{ccc}
  k[X,Y]/(I_2\cap k[X,Y])&\rightarrow&k[S]/\chi_2(S)\\
  X&\mapsto&u_2(S)\\
  Y&\mapsto&v_2(S)
\end{array}$$
we obtain that the map $k[S]/\chi_1(S)\rightarrow k[S]/\chi_2(S)$ which sends
$S$ to $\lambda_1 u_2(S)+v_2(S)$ is an isomorphism. This proves that $\chi_1$
is the characteristic polynomial of $\lambda_1 u_2(S)+v_2(S)$ in
$k[S]/\chi_2(S)$.
To prove the two congruence relations, we observe that for all 
$f\in k[X, Y]$, $f(u_1, v_1)\equiv 0\bmod \chi_1$ if and only if
$f(u_2, v_2)\equiv 0\bmod \chi_2$. In particular, the polynomial
$P(X, Y) = u_1(\lambda_1 X + Y) - X$ satisfies
$P(u_1, v_1) \equiv 0\bmod\chi_1$, and hence $P(u_2, v_2)\equiv 0\bmod \chi_2$. The
proof of the last congruence relation is similar.
\end{proof}

\section{Divisor arithmetic for smooth divisors}\label{sec:div_arith}

The first step to perform arithmetic operations on smooth divisors given by
primitive element representations is to agree on a common primitive element. In
order to achieve this, the routine \textsc{ChangePrimElt}
(Algorithm~\ref{algo:changeprim}) performs the necessary change of primitive
element by using linear algebra.  We will prove in
Propositions~\ref{prop:cost_changeprim} and \ref{prop:num_basis} that the complexity
of this step is the same as the complexity of the subroutine
\textsc{NumeratorBasis} in the main
algorithm. Therefore, decreasing the complexity of \textsc{ChangePrimElt} would
not change the global complexity and hence we make no
effort to optimize it, although it might be possible to obtain a better complexity for this step by using
a method similar to \cite[Algo.~5]{giusti2001grobner}.

Throughout this paper, for $d>0$ we let $k[S]_{<d}$ denote the vector space of
univariate polynomials with coefficients in $k$ of degree less than $d$.

\begin{algorithm} 
  \textbf{Function} \textsc{ChangePrimElt}\;
  \KwData{A scalar $\widetilde\lambda\in k$ and a primitive element
    representation $(\lambda, \chi, u, v)$ of a smooth
    effective divisor $D$.}
  \KwResult{Univariate polynomials $(\widetilde\chi, \widetilde
u, \widetilde v)$ such that $(\widetilde\lambda, \widetilde\chi, \widetilde
u, \widetilde v)$ is a primitive element representation of $D$ or ``fail''.}
\If{$\GCD(\frac{\partial q}{\partial X}(u(S),v(S))-\widetilde\lambda \frac{\partial q}{\partial Y}(u(S),
    v(S)),\chi(S)) \ne 1$}{Return ``fail''.}
  $M\gets$ $\deg(\chi)\times\deg(\chi)$ matrix representing the linear map
	$\varphi:k[S]_{<\deg(\chi)}\rightarrow k[S]_{<\deg(\chi)}$ such that
	$\varphi(f)(S)\equiv f(S)\cdot(\widetilde\lambda u(S)+ v(S))\bmod
	\chi(S)$\;
$\widetilde\chi\gets$ \textsc{CharacteristicPolynomial}($M$)\;
  $N\gets$ $\deg(\chi)\times\deg(\chi)$ invertible matrix representing the linear map
	$\psi:k[S]_{<\deg(\widetilde\chi)}\rightarrow k[S]_{<\deg(\chi)}$ such that
	$\psi(f)(S)\equiv f(\widetilde\lambda u(S)+v(S))\bmod
	\chi(S)$\;
\If{$N$ is not invertible}{Return ``fail''.}
$\widetilde u\gets \psi^{-1}(u)$\;
$\widetilde v\gets \psi^{-1}(v)$\;
Return $(\widetilde\chi, \widetilde
u, \widetilde v)$.
  \caption{Changing the primitive element in the
    representation of a smooth effective divisor.\label{algo:changeprim}}
\end{algorithm}

\begin{proposition}\label{prop:correct_changeprim}
  Algorithm~\ref{algo:changeprim} (\textsc{ChangePrimElt}) is correct: If it does not fail, then
  $(\widetilde\lambda, \widetilde\chi, \widetilde
  u, \widetilde v)$ satisfies properties \textbf{(Div-H1)} to \textbf{(Div-H3)}
  and it represents the same effective divisor as $(\lambda, \chi, 
  u,  v)$.
\end{proposition}

\begin{proof}
  First, we prove that $(\widetilde\lambda, \widetilde\chi, \widetilde
  u, \widetilde v)$ satisfies Properties \textbf{(Div-H1)} to \textbf{(Div-H3)}.
  We notice that the map $\psi$ in Algorithm~\ref{algo:changeprim} can be extended to an
  isomorphism $\Psi$ of $k$-algebras between $k[S]/\widetilde\chi(S)$ and
  $k[S]/\chi(S)$.
  Property \textbf{(Div-H1)} follows from the fact that in $k[S]/\chi(S)$, we
  have $q(\widetilde u, \widetilde
  v)=q(\Psi^{-1}(u), \Psi^{-1}(v)) = \Psi^{-1}(q(u, v))=0.$ Property
  \textbf{(Div-H2)} follows from the equalities $S =
  \Psi^{-1}(\Psi(S))=\Psi^{-1}(\widetilde\lambda u+
  v)=\widetilde\lambda\,\psi^{-1}(u(S))+\psi^{-1}(v(S)) =
  \widetilde\lambda\,\widetilde u(S)+\widetilde v(S)$ in
  $k[S]/\widetilde\chi(S)$. The fact that the equality
  $\widetilde\lambda\widetilde u(S)+\widetilde v(S) = S$ also holds in $k[S]$
  is a consequence of the degree bounds $\deg(\widetilde u), \deg(\widetilde
  v)<\deg(\widetilde\chi)$.
  If the first test does not fail, then 
 $\frac{\partial q}{\partial X}(u(S),v(S))-\widetilde\lambda \frac{\partial q}{\partial Y}(u(S),
 v(S))$ is invertible modulo $\chi(S)$. Applying $\Psi^{-1}$ shows \textbf{(Div-H3)}.

  Finally, we must prove that both representations encode the same divisor. By
  Proposition~\ref{prop:primelteq}, this amounts to show that
  $$\begin{cases}
    \widetilde\chi\text{ is the characteristic polynomial of
    }\widetilde\lambda u(S)+v(S)\\
\widetilde u(\widetilde\lambda u(S)+v(S))\equiv u(S)\bmod \chi(S)\text{ and}\\ 
\widetilde v(\widetilde\lambda u(S)+
v(S))\equiv v(S)\bmod \chi(S),
  \end{cases}$$
which is again proved directly by using the isomorphism $\Psi^{-1}$.
\end{proof}

\begin{algorithm}
 \textbf{Function} \textsc{HenselLiftingStep}\;
 \KwData{A squarefree bivariate polynomial $q\in
   k[X, Y]$, $(\lambda, \chi, u, v)$ which satisfies \textbf{(Div-H1)} to
   \textbf{(Div-H3)}, and a univariate polynomial $\widehat\chi$ which divides
 $\chi^2$.}
\KwResult{Two polynomials $\widehat u, \widehat v\in k[S]_{<\deg(\widehat\chi)}$
such that $(\lambda, \widehat\chi, \widehat u, \widehat v)$ satisfies \textbf{(Div-H1)} to 
   \textbf{(Div-H3)}.}
  $\widehat u(S)\gets\left(u(S)-\dfrac{q(u(S),
     v(S))-(\lambda u(S)+ v(S) - S)\frac{\partial q}{\partial Y}(u(S),  v(S))}{
  \frac{\partial q}{\partial X}( u(S),
 v(S))-\lambda\frac{\partial
    q}{\partial Y}( u(S),
   v(S))}\right)\bmod\widehat\chi(S)$\;
  $\widehat v(S)\gets\left( v(S)-\dfrac{-\lambda q( u(S),
     v(S))+(\lambda u(S)+ v(S) -
    S)\frac{\partial q}{\partial X}(u(S), v(S))}{
  \frac{\partial q}{\partial X}( u(S),
 v(S))-\lambda\frac{\partial
    q}{\partial Y}( u(S),
   v(S))}\right)\bmod\widehat\chi(S)$\;
Return $(\widehat u, \widehat v)$.
  \caption{A step of Newton-Hensel's lifting.\label{algo:hensel}}
\end{algorithm}

\begin{proposition}\label{prop:hensel_correct}
  Algorithm~\ref{algo:hensel} (\textsc{HenselLiftingStep}) is
  correct:
  $(\lambda, \widehat\chi, \widehat u, \widehat v)$
satisfies \textbf{(Div-H1)} to \textbf{(Div-H3)}.
\end{proposition}

\begin{proof}
  This is a special case of the Newton-Hensel's lifting. Using Taylor
  expansion, 
  $$\begin{bmatrix} q(X,Y)\\\!\lambda X+Y-S\!\end{bmatrix}\!\!=\!\!
  \begin{bmatrix} q(u(S),v(S))\\\lambda u(S)+v(S)-S\end{bmatrix}\!+\! 
  \begin{bmatrix} \frac{\partial q}{\partial X}(u(S), v(S))\!\! & \!\!\frac{\partial
    q}{\partial Y}(u(S), v(S))\\
    \lambda\!\!&\!\!1\end{bmatrix}\cdot
  \begin{bmatrix} X-u(S)\\Y-v(S)\end{bmatrix}\!+\varepsilon(X,Y,S),$$
  where $\varepsilon$ is such that $\varepsilon(\widetilde u(S),\widetilde v(S), S)\equiv
  0\bmod\chi(S)^2$ for any polynomials $\widetilde u, \widetilde v\in k[S]$ such
  that $\widetilde u\equiv u\bmod \chi$ and $\widetilde v\equiv v\bmod \chi$. Next, notice that the denominators in the definitions of
  $\widehat u$ and $\widehat v$ are invertible modulo $\chi(S)^2$ because
  they are invertible modulo $\chi(S)$.
  The proof of \textbf{(Div-H1)} and \textbf{(Div-H2)} follows from a direct computation by plugging the values of
  $\widehat u$ and $\widehat v$ in the Taylor expansion, and by noticing that
  $\widehat u\equiv u\bmod\chi$ and $\widehat v\equiv v\bmod \chi$, so that
  $\varepsilon(\widehat u(S),\widehat v(S),S)\equiv 0\bmod \chi(S)^2$ and
  hence $\varepsilon(\widehat u(S),\widehat v(S),S)\equiv 0\bmod
  \widehat\chi(S)$.
  Finally, \textbf{(Div-H3)} is a direct consequence of the fact that $\frac{\partial q}{\partial X}(u(S),v(S))-\lambda
\frac{\partial q}{\partial Y}( u(S),
     v(S))$ is invertible modulo $\chi(S)$.
\end{proof}

\begin{algorithm} 
  \textbf{Function} \textsc{AddDivisors}\;
  \KwData{A polynomial $q\in k[X,Y]$ and two smooth
    effective divisors $D_1, D_2$ given by primitive element representations $(\lambda_1, \chi_1,
    u_1, v_1)$ and $(\lambda_2, \chi_2,
    u_2, v_2)$.}
  \KwResult{A primitive element representation of the divisor $D_1+D_2$ or ``fail''.}
  $\widehat\lambda\gets$ \textsc{Random}($k$)\;
  $(\widehat\chi_1, \widehat u_1,
  \widehat v_1)\gets$ \textsc{ChangePrimElt}($\widehat\lambda, 
  \lambda_1, \chi_1, u_1, v_1$)\; 
  $(\widehat\chi_2, \widehat u_2,
  \widehat v_2)\gets$ \textsc{ChangePrimElt}($\widehat\lambda, 
  \lambda_2, \chi_2, u_2, v_2$)\; 
  \If{$\widehat u_1\not\equiv \widehat u_2\bmod
  \GCD(\widehat\chi_1, \widehat\chi_2)$}{Return ``fail''}
  $\widehat\chi\gets\widehat\chi_1\cdot\widehat\chi_2$\;
  $\widetilde\chi\gets \LCM(\widehat\chi_1,\widehat\chi_2)$\;
  $\widehat u_{12}\gets
  \XCRT((\widehat\chi_1,
  \widehat\chi_2), 
(\widehat  u_1, \widehat u_2))\in k[S]_{<\deg(\widetilde\chi)}$\;
  $\widehat v_{12}\gets
  \XCRT((\widehat\chi_1,
  \widehat\chi_2), 
(\widehat  v_1, \widehat v_2))\in k[S]_{<\deg(\widetilde\chi)}$\;
$(\widehat u, \widehat v)\gets$ \textsc{HenselLiftingStep}($q, \widetilde\chi, \widehat\lambda, \widehat
u_{12},\widehat v_{12}, \widehat\chi$)\;
  Return $(\widehat\lambda, \widehat\chi, \widehat u, \widehat v).$
  \caption{Computing the sum of two smooth effective divisors.\label{algo:sum}}
\end{algorithm}

Algorithm~\ref{algo:sum} uses a variant of the CRT, which we call the
\emph{Extended Chinese Remainder Theorem} and which we abbreviate as
$\XCRT$. Given four univariate polynomials $u_1, u_2, \chi_1,
\chi_2\in k[S]$ such that $u_1\equiv u_2\bmod\GCD(\chi_1, \chi_2)$,
it returns a polynomial $u\in k[S]$ of degree less than
$\deg(\LCM(\chi_1,\chi_2))$ such that $u\equiv u_1\bmod \chi_1$
and $u\equiv u_2\bmod \chi_2$. The main difference with the classical
CRT is that we do not require $\chi_1$ and $\chi_2$ to be coprime. A minimal solution
to the $\XCRT$ problem is given by 
\begin{equation}\label{eq:XCRT}
  \XCRT((\chi_1, \chi_2), (u_1,u_2)) = (u_2\,a_1\,(\chi_1/g) +
  u_1\,a_2\,(\chi_2/g))\bmod\LCM(\chi_1,\chi_2),
\end{equation}
where $g=\GCD(\chi_1,\chi_2)$ and $a_1,a_2\in k[S]$ are Bézout coefficients for
$\chi_1,\chi_2$, i.e. they satisfy $a_1\chi_1+a_2\chi_2 = g$.
Notice that the $\XCRT$ is in fact a $k$-algebra isomorphism between
$k[S]/\LCM(\chi_1(S),\chi_2(S))$ and the subalgebra of $k[S]/\chi_1(S) \times
k[S]/\chi_2(S)$ formed by pairs $(u_1, u_2)$ such that $u_1\equiv u_2\bmod
\GCD(\chi_1,\chi_2)$.

\begin{proposition}\label{prop:adddiv_correct}
  Algorithm~\ref{algo:sum} (\textsc{AddDivisors}) is correct: If it does not fail, then it returns a
  primitive element representation of the smooth effective divisor $D_1+D_2$.
\end{proposition}

\begin{proof} 
  Let $I_1$, $I_2$, $J$ denote the three following ideals of $k[\C^0]$:
  $$\begin{array}{rclll}
    I_1 &=& \langle \chi_1(\lambda_1 X + Y),&X-u_1(\lambda_1 X+Y),&Y-v_1(\lambda_1 X+Y)\rangle;\\
    I_2 &=& \langle \chi_2(\lambda_2 X + Y),&X-u_2(\lambda_2 X+Y),&Y-v_2(\lambda_2 X+Y)\rangle;\\
    J   &=& \langle \widehat\chi(\widehat\lambda X+ Y),&X-\widehat
    u(\widehat\lambda X+ Y),&Y-\widehat v(\widehat\lambda X+ Y)\rangle.
\end{array}$$
Proving that Algorithm~\ref{algo:sum} is correct amounts to showing that
$I_1\cdot I_2= J$, and that \textbf{(Div-H1)} to
\textbf{(Div-H3)} are satisfied by $\widehat\lambda, \widehat\chi, \widehat u, \widehat v$.
First, let $I_1', I_2'\subset k[\C^0]$ be the ideals 
  $$\begin{array}{rcl}
    I_1' &=& \langle \widehat\chi_1(\widehat\lambda X+Y), X-\widehat u_1(\widehat\lambda X+ Y), Y-\widehat v_1(\widehat\lambda X+ Y)\rangle;\\
    I_2' &=& \langle \widehat\chi_2(\widehat\lambda X+Y), X-\widehat u_2(\widehat\lambda X+ Y), Y-\widehat v_2(\widehat\lambda X+ Y)\rangle.
  \end{array}$$
  By Proposition~\ref{prop:correct_changeprim} and Lemma~\ref{lem:techrepr}, the equalities $I_1=I_1'$ and
  $I_2=I_2'$ hold.

  We start by proving that $\widehat\lambda, \widehat\chi, \widehat u, \widehat v$ satisfy \textbf{(Div-H1)} to
\textbf{(Div-H3)}. For \textbf{(Div-H1)} and \textbf{(Div-H2)},
Proposition~\ref{prop:correct_changeprim} ensures that
$q(\widehat u_i(S),\widehat v_i(S))\equiv 0\bmod \widehat\chi_i(S)$ for
$i\in\{1,2\}$. Using the fact that the $\XCRT$ is a morphism, we get that $q(\widehat u_{12}(S),\widehat v_{12}(S))\equiv
0\bmod \LCM(\widehat\chi_1(S), \widehat\chi_2(S))$ and $\widehat\lambda
\widehat u_{12}(S) + \widehat v_{12}(S)\equiv S\bmod \LCM(\widehat\chi_1(S),
\widehat\chi_2(S))$. Next, Proposition~\ref{prop:hensel_correct} proves the
equalities
$q(\widetilde u(S),\widetilde v(S))\equiv
0\bmod \LCM(\widehat\chi_1(S), \widehat\chi_2(S))^2$ and
$\widehat\lambda\widetilde u(S)+\widetilde v(S)\equiv S\bmod \LCM(\widehat\chi_1(S), \widehat\chi_2(S))^2$. Since $\widehat\chi =
\widehat\chi_1\cdot\widehat\chi_2$
divides $\LCM(\widehat\chi_1(S), \widehat\chi_2(S))^2$, we get that $q(\widetilde u(S),\widetilde v(S))\equiv
0\bmod\widehat\chi$ and $\lambda \widehat u(S)+\widehat v(S) = S$. For \textbf{(Div-H3)}, we observe that the fact
that the $\XCRT$ is a ring morphism implies that 
$\frac{\partial q}{\partial X}(\widehat u_{12}(S), \widehat v_{12}(S))-\lambda \frac{\partial
q}{\partial Y}(\widehat u_{12}(S),
\widehat v_{12}(S))$
is invertible in
$k[S]/\LCM(\widehat\chi_1(S), \widehat\chi_2(S))$. Consequently,
$\frac{\partial q}{\partial X}(\widetilde u(S), \widetilde v(S))-\lambda \frac{\partial
q}{\partial Y}(\widetilde u(S),
\widetilde v(S))$ is invertible in $k[S]/\LCM(\widehat\chi_1(S),
\widehat\chi_2(S))$, and hence it is also invertible in
$k[S]/\widehat\chi(S)$.

  We prove now that $I_1'\cdot I_2' = J$. Using the factorization as a product
  of maximal ideals given by Lemma~\ref{lem:almostdedekind}, it is sufficient to prove that a power
  $\mathfrak m^\ell\subset k[\C^0]$ of a maximal
  ideal contains $I_1'\cdot I_2'$ if and only if
  it contains $J$. Notice that the powers of maximal ideals which contain
  $I_1'$ (resp. $I_2'$) are of the form $\langle \chi_{\mathfrak m}(\widehat\lambda X+Y)^\ell,
  X-u_{1,\mathfrak m^\ell}(\widehat\lambda X+ Y), Y-v_{1,\mathfrak
  m^\ell}(\widehat\lambda X+
  Y)\rangle$ (resp. 
  $\langle \chi_{\mathfrak m}(\widehat\lambda X+Y)^\ell,
  X-u_{2,\mathfrak m^\ell}(\widehat\lambda X+ Y), Y-v_{2,\mathfrak
  m^\ell}(\widehat\lambda X+
  Y)\rangle$), where $\chi_{\mathfrak m}$ is a prime polynomial such that $\chi_{\mathfrak m}^\ell$ divides
  $\widehat\chi_1$ (resp. $\widehat\chi_2$),
  and $u_{1,\mathfrak m^\ell}(S)\equiv \widehat u_1(S)\bmod \chi_{\mathfrak
  m}(S)^\ell$, $v_{1,\mathfrak m^\ell}(S)\equiv \widehat v_1(S)\bmod \chi_{\mathfrak
  m}(S)^\ell$ (resp. $u_{2,\mathfrak m^\ell}(S)\equiv \widehat u_2(S)\bmod
  \chi_{\mathfrak
  m}(S)^\ell$, $v_{2,\mathfrak m^\ell}(S)\equiv \widehat v_2(S)\bmod \chi_{\mathfrak
  m}(S)^\ell$). 

  Let $\mathfrak m^\ell$ be a power of a maximal ideal which contains
  $I_1'\cdot I_2'$. Using the unicity of the factorization in Lemma~\ref{lem:almostdedekind}, the
  powers of maximal ideals which contain $I_1'\cdot I_2'$ are those $\mathfrak m^{\ell_1+\ell_2}$ where $I_1'\subset \mathfrak
  m^{\ell_1}$ and $I_2'\subset \mathfrak 
  m^{\ell_2}$. This means that $\mathfrak m^\ell$ has the form 
  $$\mathfrak m^\ell=\langle\chi_{\mathfrak m}(\widehat\lambda
  X+Y)^{\ell_1+\ell_2},
  X-u_{12}(\widehat\lambda X+ Y), Y-v_{12}(\widehat\lambda X+
  Y)\rangle,$$
  where $u_{12}$ (resp. $v_{12}$) is any polynomial such that $u_{12}(S) \equiv
  u_{1,\mathfrak m^{\ell_1}}(S)\bmod \chi_{\mathfrak  m}(S)^{\ell_1}$,
  $u_{12}(S) \equiv
  u_{2,\mathfrak m^{\ell_2}}(S)\bmod \chi_{\mathfrak  m}(S)^{\ell_2}$ 
  (resp.
  $v_{12}(S) \equiv
  v_{1,\mathfrak m^{\ell_1}}(S)\bmod \chi_{\mathfrak  m}(S)^{\ell_1}$,
  $v_{12}(S) \equiv
  v_{2,\mathfrak m^{\ell_2}}(S)\bmod \chi_{\mathfrak  m}(S)^{\ell_2}$).
  Then we notice that
  $\widehat\chi=\widehat\chi_1\cdot\widehat\chi_2$, and therefore
  $\chi_{\mathfrak m}(S)^{\ell_1+\ell_2}$ divides $\widehat\chi(S)$. By using the
  properties of the $\XCRT$ and of the Hensel's lifting, we get that
  $$\begin{array}{rcll}
    \widehat u(S)&\equiv&\widehat u_1(S)\bmod\widehat\chi_1(S);\\
    \widehat u(S)&\equiv&\widehat u_2(S)\bmod\widehat\chi_2(S);\\
    \widehat v(S)&\equiv&\widehat v_1(S)\bmod\widehat\chi_1(S);\\
    \widehat v(S)&\equiv&\widehat v_2(S)\bmod\widehat\chi_2(S).
  \end{array}$$
  This implies that $\mathfrak m^\ell=\langle\chi_{\mathfrak
  m}(\widehat\lambda X+Y)^{\ell_1+\ell_2}, X-\widehat
  u(\widehat\lambda X+Y), Y-\widehat 
  v(\widehat\lambda X+Y)\rangle$, and hence $\mathfrak m^\ell$ contains $J$.

  The proof that any power of maximal ideal which contains $J$ also contains
  $I_1'\cdot I_2'$ is similar.
\end{proof}

\begin{algorithm} 
  \textbf{Function} \textsc{SubtractDivisors}\;
  \KwData{Two smooth effective divisors given by primitive element
  representations:\\
    \quad\quad\quad$D_1=(\lambda_1, \chi_1,
    u_1, v_1)$, \\\quad\quad\quad$D_2 = (\lambda_2, \chi_2,
    u_2, v_2)$.}
    \KwResult{A primitive element representation of the smooth effective divisor $[D_1-D_2]_+$ or ``fail''.}
  $\widehat\lambda\gets$ \textsc{Random}($k$)\;
  $(\widehat\chi_1, \widehat u_1,
  \widehat v_1)\gets$ \textsc{ChangePrimElt}($\widehat\lambda, 
  \lambda_1, \chi_1, u_1, v_1$)\; 
  $(\widehat\chi_2, \widehat u_2,
  \widehat v_2)\gets$ \textsc{ChangePrimElt}($\widehat\lambda, 
  \lambda_2, \chi_2, u_2, v_2$)\;
  \If{$\widehat u_1\not\equiv \widehat u_2\bmod
  \GCD(\widehat\chi_1, \widehat\chi_2)$}{Return ``fail''}
  $\widehat\chi\gets\widehat\chi_1/\GCD(\widehat\chi_1,\widehat\chi_2)$\;
  $\widehat u(S)\gets \widehat u_1(S)\bmod \widehat\chi(S)$\;
  $\widehat v(S)\gets \widehat v_1(S)\bmod \widehat\chi(S)$\;
  Return $(\widehat\lambda, 
  \widehat\chi, \widehat u,
  \widehat v).$
  \caption{Computing the subtraction of smooth effective divisors.\label{algo:subtract}}
\end{algorithm}

Algorithm~\ref{algo:subtract} (\textsc{SubtractDivisors}) provides a method for subtracting effective
divisors given by primitive element representations. We emphasize that the
divisor returned is the subtraction $D_1-D_2$ only if the result is also
effective, i.e. if $D_1\geq D_2$. If this is not the case, then it returns
the positive part of the subtraction.

\begin{proposition}
  Algorithm~\ref{algo:subtract} (\textsc{SubtractDivisors}) is correct: If it does not fail, then it returns a
  primitive element representation of the smooth effective divisor $[D_1-D_2]_+$, where the notation
  $[D]_+$ denotes the positive part of the divisor $D$, i.e. the smallest
  effective divisor $D'$ such that $D'\geq D$.
\end{proposition}

\begin{proof}
  Let $I_1$, $I_2$, $J$ denote the three following ideals of $k[\C^0]$, using
  the notation in Algorithm~\ref{algo:subtract}:
  $$\begin{array}{rclll}
    I_1 &=& \langle \chi_1(\lambda_1 X + Y),&X-u_1(\lambda_1 X+Y),&Y-v_1(\lambda_1 X+Y)\rangle;\\
    I_2 &=& \langle \chi_2(\lambda_2 X + Y),&X-u_2(\lambda_2 X+Y),&Y-v_2(\lambda_2 X+Y)\rangle;\\
    J   &=& \langle \widehat\chi(\widehat\lambda X+ Y),&X-\widehat
    u(\widehat\lambda X+ Y),&Y-\widehat v(\widehat\lambda X+ Y)\rangle.
\end{array}$$
  The effective divisor $[D_1-D_2]_+$ corresponds to the colon ideal $I_1 : I_2
  = \{f\in k[\C^0] \mid f\cdot I_2\subset I_1\}$. Consequently, we must prove that
  $(\widehat\lambda, \widehat\chi, \widehat u, \widehat v)$ satisfies \textbf{(Div-H1)} to \textbf{(Div-H3)} and that $J = I_1: I_2$.
 The
 equalities \textbf{(Div-H1)} to \textbf{(Div-H3)} for $\widehat\lambda,
 \widehat\chi_1, \widehat u_1, \widehat v_1$ are satisfied by
 Proposition~\ref{prop:correct_changeprim}. Regarding them modulo
 $\widehat\chi$ shows that $(\widehat\lambda, \widehat\chi, \widehat u, \widehat v)$ 
 satisfies
 \textbf{(Div-H1)} to \textbf{(Div-H3)}.

  In order to prove that $J = I_1 : I_2$, we proceed as in the proof of Proposition~\ref{prop:adddiv_correct}, by
  noticing first that $I_1$ and $I_2$ can be rewritten as
  $$\begin{array}{rclll}
    I_1 &=& \langle \widehat\chi_1(\widehat\lambda X + Y),&X-\widehat u_1(\widehat\lambda
    X+Y),&Y-\widehat v_1(\widehat\lambda X+Y)\rangle;\\
    I_2 &=& \langle \widehat\chi_2(\widehat\lambda X + Y),&X-\widehat u_2(\widehat\lambda
    X+Y),&Y-\widehat v_2(\widehat\lambda X+Y)\rangle.\\
\end{array}$$
Using~\cite[Prop.~9.1]{atiyah1969introduction} together with the fact that $D_1$ does not
involve any singular point of the curve by \textbf{(Div-H3)}, the equality $I_1:I_2 = J$ holds if and only if the
powers of maximal ideals $\mathfrak m^\ell\subset k[\C^0]$ which contain
$I_1:I_2$ are exactly those which contain $J$. Equivalently, this means that if
$\mathfrak m^{\ell_1}$ is the largest power of $\mathfrak m$ which contains
$I_1$ and if $\mathfrak m^{\ell_2}$ is the largest power of $\mathfrak m$ which contains
$I_2$, then $\mathfrak m^{\max(\ell_1-\ell_2, 0)}$ is the largest power of
$\mathfrak m$ which contains $J$. As in the proof of
Proposition~\ref{prop:adddiv_correct}, the maximal ideals $\mathfrak m\subset k[\C^0]$ which
contain $I_1$ have the form $\langle \chi_{\mathfrak m}(\widehat\lambda X +
Y),X-u_{\mathfrak m}(\widehat\lambda
X+Y),Y-v_{\mathfrak m}(\widehat\lambda X+Y)\rangle$, where $u_{\mathfrak
m}\equiv \widehat u_1\bmod \chi_{\mathfrak m}, v_{\mathfrak
m}\equiv \widehat v_1\bmod \chi_{\mathfrak m}$. The proof is concluded by
noticing that for any
prime factor $\Phi$ of $\widehat\chi_1$, if $\Phi^{\ell_1}$ is the largest
power of $\Phi$ which divides $\widehat\chi_1$ and $\Phi^{\ell_2}$ is the largest
power of $\Phi$ which divides $\widehat\chi_2$, then the largest power $\Phi$
which divides $\widehat\chi=\widehat\chi_1/\GCD(\widehat\chi_1,\widehat\chi_2)$
is $\Phi^{\max(\ell_1-\ell_2, 0)}$.
\end{proof}

\section{Description and correctness of the subroutines}\label{sec:subroutines}

\subsection{Interpolation}

\begin{algorithm}
 \textbf{Function} \textsc{Interpolate}\;
 \KwData{The degree $\delta$ of the curve, a smooth effective divisor
 given by a primitive element representation $(\lambda, \chi, u, v)$, and the
 nodal divisor given by $(\lambda_E, \chi_E, u_E, v_E, T_E)$.}
 \KwResult{A polynomial $h\in k[X, Y]$ representing a form in $k[\C]$ such that
 $(h)\geq D+E$.}
 \eIf{$\displaystyle\binom{\delta+1}2\leq\deg(\chi)+\deg(\chi_E)$}{
 $d\gets \lfloor(\deg(\chi)+\deg(\chi_E))/\delta+(\delta-1)/2\rfloor$}{
   $d\gets \lfloor(\sqrt{1+8(\deg(\chi)+\deg(\chi_E))}-1)/2\rfloor$}
   Construct the matrix representing the linear map $\varphi:\{f\in k[X,Y]\mid
   \deg(f)\leq d, \deg_Y(f)<\delta\}\rightarrow k[S]_{<\deg(\chi)}\times
   k[S]_{<\deg(\chi_E)}$
   defined as $\varphi(f(X, Y))=(f(u(S), v(S))\bmod \chi(S), f(u_E(S), v_E(S))\bmod
   \chi_E(S))$\;
   Compute a basis $\mathbf b_1,\ldots, \mathbf b_\ell$ of the kernel of $\varphi$\;
   $(\mu_1,\ldots, \mu_\ell)\gets \textsc{Random}(k^\ell\setminus\{\mathbf
   0\})$\;
   Return $h=\sum_{i=1}^\ell \mu_i \mathbf b_i$.
   \caption{Computing a function $h\in k[\C^0]$ of small degree such that $(h)\geq
   D+E$.\label{algo:interpolation}}
\end{algorithm}

This section focuses on the following interpolation problem: Given a smooth effective
divisor $D$ and the nodal divisor $E$, find a element $h\in k[\C^0]$ such that its associated
principal divisor $(h)$ satisfies $(h)\geq D+E$.

\begin{proposition}\label{prop:interp_correct}
  Algorithm~\ref{algo:interpolation} (\textsc{Interpolate}) is correct: The kernel of
  $\varphi$ has positive dimension, and its nonzero elements $h$ satisfy
  $(h)\geq D+E$.
\end{proposition}

\begin{proof}
  The fact that the kernel $\varphi$ has positive dimension follows from a
  dimension count, which is postponed to Lemma~\ref{lem:varphinotinj}.
  We now prove the second part of the proposition. First, notice that
  $\deg_Y(h)<\deg_Y(q)$ for any nonzero $h$ in the kernel of $\varphi$, hence
  $h$ cannot be a multiple of $q$, which implies that
  $\langle 0\rangle\subsetneq\langle h\rangle\subset k[\C^0]$. Next, by
Lemmas~\ref{lem:divHprimelt} and \ref{lem:techrepr}, the ideal
$I_{D+E}= \{f\in k[\C^0] \mid (f)\geq D+E\}=\{f\in k[\C^0] \mid (f)\geq D\}\cap \{f\in k[\C^0] \mid (f)\geq E\}$ equals $$\langle \chi(\lambda X+Y), X-u(\lambda X+Y), Y-v(\lambda
  X+Y)\rangle\cap \langle \chi_E(\lambda_E X+Y), X-u_E(\lambda_E X+Y), Y-v_E(\lambda
  X+Y)\rangle.$$
  By
  construction, $h(u(S), v(S))\equiv
  0\bmod\chi(S)$ and $h(u_E(S), v_E(S))\equiv
  0\bmod\chi_E(S)$ for any $h\in\ker\varphi$.
  The proof is concluded by noticing that the polynomials $f$ in $I_{D+E}$ are exactly those which satisfy $f(u(S), v(S))\equiv
  0\bmod\chi(S)$ and $f(u_E(S), v_E(S))\equiv
  0\bmod\chi_E(S)$, using the isomorphisms in
  Propositions~\ref{prop:existence_primrepr}
  and~\ref{prop:existence_primrepr_nodal}. \end{proof}

The following lemma ensures that Algorithm~\ref{algo:interpolation} actually
returns a nonzero element, i.e. that the kernel of $\varphi$ has positive
dimension.

\begin{lemma}\label{lem:varphinotinj}
  With the notation in Algorithm~\ref{algo:interpolation}, 
$$\deg(\chi)+\deg(\chi_E)<\dim_k(\{f\in k[X, Y]\mid
   \deg(f)\leq d, \deg_Y(f)<\delta\})\leq 3(\deg(\chi)+\deg(\chi_E)).$$
   Consequently, $\varphi$ is not injective.
\end{lemma}

\begin{proof}
  Set $w=\deg(\chi)+\deg(\chi_E)$.
First, a direct dimension count gives
$$\dim_k(\{f\in k[X, Y]\mid
   \deg(f)\leq d, \deg_Y(f)<\delta\}) =\begin{cases}
    \delta(d-(\delta-3)/2)\text{ if }d \geq \delta\\
    \displaystyle\binom{d+2}{2}\text{ otherwise}.
   \end{cases}$$
 On one hand, if $\binom{\delta+1}{2}\leq w$, then 
 $$\begin{array}{rcl}
   d &=& \displaystyle
   \left\lfloor w/\delta+(\delta-1)/2\right\rfloor\\
  &\geq&
  \left\lfloor \binom{\delta+1}{2}/\delta+(\delta-1)/2\right\rfloor\\
  &\geq&\delta,\\
\end{array}$$
and hence
$$\begin{array}{rcl}
\delta(d-(\delta-3)/2)&>&\delta(w/\delta+(\delta-1)/2-1-(\delta-3)/2)\\
&=&w\\
\delta(d-(\delta-3)/2)&\leq&\delta(w/\delta+(\delta-1)/2-(\delta-3)/2)\\
&\leq&w+\delta\\
&\leq&w+\binom{\delta+1}2\\
&\leq&2w.
\end{array}$$
On the other hand, if $\binom{\delta+1}{2}> w$, then
$$\begin{array}{rcl}
  d &=& \lfloor(\sqrt{1+8w}-1)/2\rfloor\\
    &<& \lfloor(\sqrt{1+4\delta(\delta+1)}-1)/2\rfloor \\
    &=& \lfloor(\sqrt{(2\delta+1)^2}-1)/2 \rfloor \\
    &=& \delta
  \end{array}$$
  Since $\binom{x+2}{2} - w>0$ for any $x >
  (\sqrt{1+8w}-3)/2$, we get that $w < \binom{d+2}{2}$ as
  expected. Finally, the last inequality follows from
  $$
    \displaystyle\binom{\lfloor(\sqrt{1+8w}-1)/2\rfloor+2}2\leq w+(1+\sqrt{1+8w})/2,
  $$
  and direct computations show that $(1+\sqrt{1+8w})/2\leq
  2w$.
\end{proof}

\subsection{Computing the smooth part of the principal divisor associated to a regular function on
the curve}

The section is devoted to the following problem: Given a polynomial $h\in
k[\C^0]$ such that $(h)=D_h+E$ where $D_h$ is a smooth divisor on the curve,
compute a primitive element representation of $D_h$. 

Let us mention that it may happen that $h$ vanishes at
infinity. Therefore, the support of $D_h$ may contain points at infinity, but
the primitive element representation only represents points in the affine chart
$Z\ne 0$.
Ignoring these zeros at infinity may lead to functions
having unauthorized poles at infinity in the basis returned by
Algorithm~\ref{algo:birdeye}.
As we already mentioned in Section~\ref{sec:datastruct}, handling what
happens at infinity is not a
problem: This issue can be solved for instance by doing the computations in
three affine spaces which cover $\mathbb P^2$, which would multiply the
complexity by a constant factor.
Notice also that it is easy to detect if $h$ has zeroes at infinity: This
happens if and only if the degree of the resultant of $h$ and $q$
is strictly less than $\deg(h)\deg(\C)$, thanks to the fact
that we assumed that $\C$ is in projective Noether position.
For simplicity, we
will not discuss further this issue in the sequel of this paper. 

The central element of Algorithm~\textsc{CompPrincDiv} is the computation of
a resultant and of the associated first subresultant (as defined for instance
in~\cite[Sec.~3]{kahoui2003elementary}). However, a number of extra steps are required 
to ensure that this computation satisfies genericity assumptions and returns a
correct result. First, a random direction of projection $\lambda$ is selected for
computing the resultant. This direction of projection must satisfy some
conditions. In particular, distinct point in the support of $h$ must project on
distinct points. Also, in order to exploit the Poisson formula for the
resultant, we also ask that this direction is not a tangent at any node of the
curve, see Lemma~\ref{lem:tech_res}. This condition about the tangents at the
node is tested via the evaluation of the univariate polynomial $T_E$. We also need a representation of the nodal
divisor with respect to this $\lambda$. This is achieved by using a
slightly modified version of Algorithm~\textsc{ChangePrimElt} where the first
test --- which is not relevant for the nodal divisor --- is removed. Finally,
Algorithm~\textsc{CompPrincDiv} must clean out the singular points: this is
done by noticing that the roots of the resultant which correspond to the
singular points appear with multiplicity at least $2$, see
Lemma~\ref{lem:tech_res} below. Therefore these singular points are removed by dividing
out by the
square of the univariate polynomial $\widehat\chi_E$ whose roots parametrize the coordinates of the
singular points.

\begin{algorithm}
 \textbf{Function} \textsc{CompPrincDiv}\;
 \KwData{A squarefree bivariate $q\in k[X, Y]$ such that
   $\deg(q)=\deg_Y(q)$, a bivariate polynomial $h\in k[X, Y]$, and a
 representation $(\lambda_E, \chi_E, u_E, v_E, T_E)$ of the nodal divisor.}
 \KwResult{A primitive element representation $(\lambda, \chi(S), u(S), v(S))$
 of the smooth part of the principal effective divisor $(h)$ or ``fail''.}
 $\lambda\gets$ \textsc{Random}($k$)\;
 \If{$\lambda = 0$ or if the coefficient of $Y^{\deg(q)}$ in $q((S-Y)/\lambda,
 Y)\in k[S][Y]$ is $0$}{Return ``fail''}
 \If{$T_E(\lambda)=0$}{Return ``fail''\; \tcc{Ensures that $(1, -\lambda)$ is not tangent to the curve
 at any node.}}
 $(\widehat\chi_E, \widehat u_E, \widehat v_E)\gets$
 \textsc{ChangePrimEltNodal}($\lambda$, $\lambda_E$, $\chi_E$, $u_E$, $v_E$)\;
 \tcc{\textsc{ChangePrimEltNodal} is the same algorithm as \textsc{ChangePrimElt}, 
 but we
 skip the first test (which would fail on the nodal divisor).}
 $\widetilde\chi(S)\gets$ $\Resultant_Y$($q((S-Y)/\lambda, Y), h(
 (S-Y)/\lambda, Y)$)\;
 $a_0(S)+Y a_1(S)\gets$ $\FirstSubRes_Y$($q((S-Y)/\lambda, Y), h(
 (S-Y)/\lambda, Y)$)\;
 $\chi\gets \widetilde\chi/ \widehat\chi_E^2$\;
 \If{$\GCD(\chi, \widehat \chi_E)\ne 1$}{Return ``fail''}
 \If{$\GCD(a_1(S), \chi(S))\ne 1$}{Return ``fail''}
 $v(S)\gets -a_0(S)\cdot a_1(S)^{-1}\bmod \chi(S)$\;
 $u(S)\gets (S-v(S))/\lambda$\;
 \If{$\GCD(\frac{\partial q}{\partial X}(u(S), v(S))-\lambda \frac{\partial q}{\partial Y}(u(S),
    v(S)), \chi(S))\ne 1$}{Return ``fail''}
 Return $(\lambda, \chi(S), u(S), v(S))$.
 \caption{Computing a primitive element representation of the smooth
   part of $(h)$.\label{algo:princdiv}}
\end{algorithm}

\begin{proposition}\label{prop:correct_princdiv}
  Algorithm~\ref{algo:princdiv} (\textsc{CompPrincDiv}) is correct: If it does not fail, then it
  returns a primitive element representation of the smooth part of the principal divisor $(h)$.
\end{proposition}

Before proving Proposition~\ref{prop:correct_princdiv}, we need the
following technical lemma, which implies in particular that nodes
appear as roots of the resultant with multiplicity at least two.

\begin{lemma}\label{lem:tech_res}
  With the notation in Algorithm~\ref{algo:princdiv}, 
  let $s\in \overline k$ and $\lambda\in k\setminus\{0\}$. Let $R_1,\ldots, R_\ell$ be
  the valuation rings in $\Frac(\overline k[X,Y]/q)$ associated
  to the points of $\widetilde C$ above $s$, i.e. the
  points on $\mathcal C$ which project to $s$ via the projection $(X,Y)\mapsto
  \lambda X+Y$. Assume that the vector
  $(1,-\lambda)$ is not tangent to
  $\mathcal C^0$ at any of these points and that the coefficient of $Y^{\deg(\C)}$ in 
  $q( (S-Y)/\lambda, Y)$ is nonzero.
  Let $m_1,\ldots, m_{\deg(\C)}$ denote the valuations of $h$ in these valuation
  rings.
  Then $s$ is a root
  of multiplicity $\sum_{i=1}^\ell m_i$ in $\Resultant_Y(q( (S-Y)/\lambda, Y),
  h( (S-Y)/\lambda, Y))$.
\end{lemma}

\begin{proof}
  Since we assumed that the curve is nodal, that the coefficient of $Y^{\deg(\C)}$ in $q( (S-Y)/\lambda,
  Y)$ is nonzero, and that the vector $(1, -\lambda)$ is not tangent at any
  point above $s$, we get that  the polynomial $q( (S-Y)/\lambda,
  Y)$ splits over the ring $\overline k[[S-s]]$ of power series at $s$ as a
  product of $\deg(\C)$ factors, see e.g.~\cite{neiger2017fast} and references
  therein. Notice that this factorization property holds even if some of the
  points above $s$ are nodes. Let $\widetilde y_1, \ldots,\widetilde y_{\deg(\C)}$ denote
  its roots in $\overline k[[S-s]]$.
  Using the multiplicativity property of the
  resultant \cite[Sec.~5.7]{jouanolou1991formalisme}, we get
  $$\Resultant_Y(q( (S-Y)/\lambda, Y), h( (S-Y)/\lambda, Y)) =
  \alpha \prod_{i=1}^\ell h( (S-\widetilde
  y_i)/\lambda, \widetilde y_i),$$
  where $\alpha\in k$.
  The proof is concluded by noticing that $S-s$ is a
  uniformizing element for all the discrete valuation rings since the vector
  $(1,-\lambda)$ is not tangent to the curve at any of the points above $s$, so
  that $m_i$ precisely
  corresponds to the largest integer $\gamma$ such that $(S-s)^\gamma$ divides
  $h((s-\widetilde y_i)/\lambda, \widetilde y_i)$.
\end{proof}

\begin{proof}[Proof of Proposition~\ref{prop:correct_princdiv}]
  In order to prove Proposition~\ref{prop:correct_princdiv}, we must prove that
  the output $(\lambda,
  \chi, u,v)$ satisfies \textbf{(Div-H1)} to \textbf{(Div-H3)} and that the two
  ideals $\langle h\rangle:I_E^\infty\subset k[\C^0]$ and $\langle \chi(\lambda X+Y),
  X-u(\lambda X+Y), Y-v(\lambda X+Y)\rangle\subset k[\C^0]$ are
  equal, where $I_E$ is the radical ideal of $k[\C^0]$ which encodes the
  algebraic set of the nodes. \textbf{(Div-H2)} follows directly from the definitions of $u(S)$
  and $v(S)$
  in Algorithm~\ref{algo:princdiv}. To prove \textbf{(Div-H1)}, we shall prove
  that the equality holds modulo $(S-s)^\gamma$ for any root $s\in\overline k$ of
  $\chi$ of multiplicity $\gamma$. A classical property of the subresultants is
  that they belong to the ideal generated by the input polynomials. This
  implies that for any root $s\in\overline k$ of $\widetilde\chi$ we have
  $$ a_0(S)+Y a_1(S)\in \langle q( (S-Y)/\lambda, Y),h( (S-Y)/\lambda,
  Y)\rangle\subset \overline k[[S-s]][Y].$$
  If the algorithm does not fail, then $a_1(S)$ is invertible modulo $\chi(S)$.
  Consequently, it is also invertible in
  $\overline k[[S-s]]$ for any root $s\in\overline k$ of $\chi$ and hence 
  $$ Y+a_0(S)a_1(S)^{-1}\in \langle q( (S-Y)/\lambda, Y),h( (S-Y)/\lambda,
  Y)\rangle\subset \overline k[[S-s]][Y].$$
  Therefore, the GCD of $q( (S-Y)/\lambda,
  Y)$ and $h( (S-Y)/\lambda,
  Y)$ in $\Frac(\overline k[[S-s]])[Y]$ divides $Y+a_0(S)a_1(S)^{-1}$. But we also know that this GCD is
  nonconstant, since $s$ is a root of the resultant $\widetilde\chi$. By a
  degree argument, this GCD
  equals $Y+a_0(S)a_1(S)^{-1}$ and hence $q(
  (S+a_0(S)a_1(S)^{-1})/\lambda, -a_0(S)a_1(S)^{-1}) = 0$ in $\overline
  k[[S-s]]$. Considering this equation modulo $(S-s)^\gamma$ and using the
  CRT over all the roots of $\chi$ finishes the proof of \textbf{(Div-H1)}.
  Finally, \textbf{(Div-H3)} is explicitely tested and hence it must be
  satisfied if the algorithm does not fail.

  It remains to prove the equality of the ideals $\langle h\rangle:I_E^\infty\subset
  k[\C^0]$ and $\langle \chi(\lambda X+Y), X-u(\lambda X+Y), Y-v(\lambda
  X+Y)\rangle\subset k[\C^0]$. Using the isomorphism between $k[X,Y]/\langle
  \chi(\lambda X+Y), X-u(\lambda X+Y), Y-v(\lambda X+Y)\rangle$ and
  $k[S]/\chi(S)$ (see Proposition~\ref{prop:existence_primrepr}), the elements
  in $\langle \chi(\lambda X+Y),
  X-u(\lambda X+Y), Y-v(\lambda X+Y)\rangle$ are precisely the classes of the
  bivariate polynomials $\psi(X, Y)\in k[X,Y]$ such that $\psi(u(S),
  v(S))\equiv 0\bmod \chi(S)$. Using a proof identical to that of
  \textbf{(Div-H1)} we get that $h(u(S), v(S))\equiv 0\bmod\chi(S)$ which
  proves that $\langle h\rangle\subset\langle \chi(\lambda X+Y), X-u(\lambda X+Y), Y-v(\lambda
  X+Y)\rangle$. Saturating on both sides, we get that $\langle
  h\rangle:I_E^\infty\subset\langle \chi(\lambda X+Y), X-u(\lambda X+Y), Y-v(\lambda
  X+Y)\rangle:I_E^\infty = \langle \chi(\lambda X+Y), X-u(\lambda X+Y), Y-v(\lambda
  X+Y)\rangle$, where the last equality comes from the fact that
  $\GCD(\chi, \widehat\chi_E) = 1$.  For the other inclusion, we
  use~\cite[Prop.~9.1]{atiyah1969introduction}, which implies that
  $\langle \chi(\lambda X+Y), X-u(\lambda X+Y),
  Y-v(\lambda X+Y)\rangle\subset \langle h\rangle:I_E^\infty$ if this inclusion holds in
  the local ring associated to any maximal ideal $\mathfrak m\subset \overline k[\C^0]$ which contains
  $\langle \chi(\lambda X+Y), X-u(\lambda X+Y), Y-v(\lambda X+Y)\rangle$. Over
  $\overline k$, these maximal ideals have the form $\langle
 \lambda X+Y - s, X-u(s), Y-v(s)\rangle$, where $s\in\overline k$ is a root of
 $\chi$. The assumption $\GCD(\frac{\partial q}{\partial X}(u(S), v(S))-\lambda\frac{\partial q}{\partial Y}(u(S), v(S)),
  \chi(S))= 1$ ensures that all these
 maximal ideals correspond to nonsingular points, and hence the associated local rings
 are discrete valuation rings. 
 For $s\in\overline k$ a root of $\chi$, let $y_1,\ldots, y_{\deg(\C)}$ be the roots
 of the univariate polynomial $q( (s-Y)/\lambda, Y)\in\overline k[Y]$. Let $m_i$ denote the intersection multiplicity of
 $h$ at the point $( (s-y_i)/\lambda, y_i)$ of $\C^0$.
Since $\GCD(\frac{\partial q}{\partial X}(u(S), v(S))-\lambda \frac{\partial q}{\partial Y}(u(S),
    v(S)), \chi(S))\ne 1$, we obtain that the vector $(1, -\lambda)$ is not
    tangent to $\C^0$ any of these points.
    Lemma~\ref{lem:tech_res} then gives that $m_1+\dots+m_{\deg(\C)}=
  \alpha$, where $\alpha$ is the multiplicity of the root $s$ in $\chi$. Let
  $k$ be the integer such that $y_k = v(s)$. Then $m_k\leq \alpha$, which
  shows that we have $\langle \chi(\lambda X+Y), X-u(\lambda X+Y), Y-v(\lambda
  X+Y)\rangle\subset \langle h\rangle:I_E^\infty$ in the local ring at the point $(u(s),
  v(s))$. The statement~\cite[Prop.~9.1]{atiyah1969introduction} concludes the proof of the inclusion $\langle
  \chi(\lambda X+Y), X-u(\lambda X+Y), Y-v(\lambda
  X+Y)\rangle\subset\langle h\rangle:I_E^\infty$.
\end{proof}

\subsection{Computing the linear space of regular functions of bounded degree having prescribed
zeros}

The task accomplished by Algorithm~\textsc{NumeratorBasis} is similar to what
Algorithm~\textsc{Interpolate} does: It computes a basis of the vector space of
regular functions having prescribed zeros. The only difference with
Algorithm~\textsc{Interpolate} is that Algorithm~\textsc{NumeratorBasis}
returns a basis of this linear space.

\begin{algorithm}
 \textbf{Function} \textsc{NumeratorBasis}\;
 \KwData{A positive integer $\delta$, a smooth effective divisor
 given by a primitive element representation $(\lambda, \chi(S), u(S), v(S))$, a positive integer $d$, and the
 nodal divisor given by $(\lambda_E, \chi_E, u_E, v_E, T_E)$.}
 \KwResult{A basis of the space of polynomials $g\in k[X, Y]$ such that
   $\deg(g)\leq d, \deg_Y(g)<\delta$ and the associated divisor
 satisfies $(g)\geq D+E$.}
   Construct the matrix representing the linear map $\varphi:\{f\in k[X,Y]\mid
   \deg(f)\leq d, \deg_Y(f)\leq\delta\}\rightarrow k[S]_{<\deg(\chi)}\times
     k[S]_{<\deg(\chi_E)}$
   defined as $\varphi(f(X, Y))=(f(u(S), v(S))\bmod \chi(S), f(u_E(S),
   v_E(S))\bmod\chi_E(S))$\;
   Compute and return a basis of the kernel of $\varphi$.
   \caption{Computing a basis of the vector space of regular functions $g\in
     k[\C^0]$ of degree $\delta$
     such that $(g)\geq D+E$.\label{algo:numbasis}}
\end{algorithm}

\begin{proposition}
  Algorithm~\ref{algo:numbasis} (\textsc{NumeratorBasis}) is correct: 
 the nonzero elements $g$ in the kernel of $\varphi$ are not divisible
 by $q$ and they satisfy $(g)\geq D+E$.
\end{proposition}

\begin{proof}
  The proof is similar to that of Proposition~\ref{prop:interp_correct}.
\end{proof}

\section{Complexity}\label{sec:complexity}
All complexity bounds count the number of arithmetic operations (additions,
subtractions, multiplications, divisions) in $k$, all at unit cost. We do not
include in our complexity bounds the cost of generating random elements, nor the
cost of monomial manipulations, nor multiplications by fixed integer constants. In
particular, we do not include in our complexity bounds the cost of computing
the partial derivatives of a polynomial. We use the classical $O()$ and
$\widetilde O()$ notation, see e.g.~\cite[Sec.~25.7]{von2013modern}. The notation $\Mcomp(n)$ stands
for the number of arithmetic operations required in $k$ to compute the product
of two univariate polynomials of degree $n$ with coefficients in $k$. By~\cite{cantor1991fast}, $\Mcomp(n)=O(n \log n\log\log n)$. In the
sequel, $\omega$ is a feasible exponent for matrix multiplication, i.e.
$\omega$ is
such that there is an algorithm for multiplying two $N\times N$ matrices with
entries in $k$ within $O(N^\omega)$ arithmetic operations in $k$. The best
known bound is $\omega < 2.3729$ \cite{le2014powers}. In the following, we make
the assumption\footnote{If $\omega=2$, then the $O()$ in
  Theorem~\ref{thm:compl} should be replaced by $\widetilde O()$.} that $\omega>2$.

\begin{proposition}\label{prop:cost_changeprim}
  Algorithm~\ref{algo:changeprim} (\textsc{ChangePrimElt}) requires at most
  $O(\deg(\chi)^\omega)$
  arithmetic operations in $k$.
\end{proposition}

\begin{proof}
  In order to construct the matrix $M$ in Algorithm~\ref{algo:changeprim}, we
  must compute the remainders
  $S^i\cdot(\widetilde\lambda u(S)+v(S))\bmod\chi(S)$ for $i\in\{0,\ldots,
  \deg(\chi)-1\}$. Each of these computations costs $O(\Mcomp(\deg(\chi)))$
  arithmetic operations, so the total cost of constructing the matrix $M$ is
  bounded by $O(\deg(\chi)\Mcomp(\deg(\chi)))$, which is bounded above
  by $O(\deg(\chi)^\omega)$. Computing the characteristic polynomial of $M$ can
  be done within $O(\deg(\chi)^\omega)$ arithmetic operations~\cite{pernet2007faster}. 
  We emphasize that in~\cite{pernet2007faster}, it is assumed that the cardinality of $k$ is
  at least $2 \deg(\chi)^2$, so that the probability of failure is bounded by
  $1/2$. In fact, using the same algorithm and the same proof as in~\cite{pernet2007faster}, the assumption on the
  cardinality of $k$ can be removed but the probability of failure will then
  only
  be bounded by $\deg(\chi)^2/\lvert \mathcal E\rvert$, where $\mathcal E\subset k$ is a finite
  subset in which we can draw elements uniformly at random. We will incorporate
  this probability of failure for the computation of
  the characteristic polynomial in our bound for the probability of failure of
  the main algorithm, see the proof of Theorem~\ref{thm:glob_bound_proba}.

  Constructing the matrix $N$ is done by computing successively the
  remainders $(\widetilde\lambda u(S)+v(S))^i\bmod \chi(S)$ for $i\in\{0,\ldots,
  \deg(\chi)-1\}$ at a total cost of $O(\deg(\chi)\Mcomp(\deg(\chi)))$ which is
  again bounded by $O(\deg(\chi)^\omega)$. Finally, inverting $N$ and applying
  the inverse linear map can be done using $O(\deg(\chi)^\omega)$ operations in
  $k$ by using \cite{bunch1974triangular}.
\end{proof}

\begin{proposition}\label{prop:cost_hensel}
  Algorithm~\ref{algo:hensel} (\textsc{HenselLiftingStep}) requires at most
  $O(\deg(q)^2\Mcomp(\deg(\chi)))$
  arithmetic operations in $k$.
\end{proposition}

\begin{proof}
  Algorithm~\ref{algo:hensel} consists in evaluations of $q$ and its partial derivatives at
  $(u(S), v(S))$, together with finitely many arithmetic operations
  in $k[S]/\chi(S)^2$ . Each of the arithmetic operations modulo $\chi^2$ costs
  $O(\Mcomp(\deg(\chi)))$ arithmetic operations in $k$. Evaluating $q$ at $(u(S), v(S))$ modulo $\chi(S)^2$
  can be done by computing the remainders $u(S)^iv(S)^j\bmod\chi(S)^2$ for all
  $(i, j)\in \Z_{\geq 0}$ such that $i+j\leq \deg(q)$, then by multiplying
  these evaluations by the corresponding coefficients in $q$ and by summing
  them. Computing all the modular products can be done in
  $O(\deg(q)^2\Mcomp(\deg(\chi)))$ operations in $k$, by considering the pairs
  $(i,j)$ in increasing lexicographical ordering. Multiplying by the
  coefficients and summing then costs $O(\deg(q)^2\deg(\chi))$ arithmetic
  operations in $k$. Computing the evaluations of the partial derivatives of
  $q$ is done similarly and it has a similar cost.
\end{proof}

\begin{proposition}\label{prop:cost_sum}
  Algorithm~\ref{algo:sum} (\textsc{AddDivisors}) requires at most
  $O(\deg(q)^2\Mcomp(\nu) +
  \nu^\omega)$
  arithmetic operations in $k$, where $\nu=\max(\deg(\chi_1), \deg(\chi_2))$.
\end{proposition}

\begin{proof}
  Algorithm~\ref{algo:sum} starts by two calls to
  the function \textsc{ChangePrimElt}, with respective costs
  $O(\deg(\chi_1)^\omega)$ and $O(\deg(\chi_2)^\omega)$ by
  Proposition~\ref{prop:cost_changeprim}. The polynomial
  $\GCD(\widehat\chi_1, \widehat\chi_2)$ can be computed at cost
  $O(\Mcomp(\nu)\log(\nu))$ using the fast GCD
  algorithm~\cite[Coro.~11.9]{von2013modern}. The product $\widehat\chi$ in Algorithm~\ref{algo:sum} and the LCM
  are then also computed at costs $O(\Mcomp(\nu))$ and $O(\Mcomp(\nu)\log(\nu))$. The
  $\XCRT$ can be computed at cost $
  O(\Mcomp(\nu)\log(\nu))$ by using Equation~\eqref{eq:XCRT} together with
  the fact that Bézout coefficients can be computed within quasi-linear
  complexity~\cite[Coro.~11.9]{von2013modern}. 
  Finally, the
  Hensel lifting step can be achieved at cost
  $O(\deg(q)^2\Mcomp(\nu))$ by
  Proposition~\ref{prop:cost_hensel}.
\end{proof}

\begin{proposition}\label{prop:cost_sub}
  Algorithm~\ref{algo:subtract} (\textsc{SubtractDivisors}) requires at most
  $O(\nu^\omega)$
  arithmetic operations in $k$, where $\nu=\max(\deg(\chi_1),\deg(\chi_2))$.
\end{proposition}

\begin{proof}
  Most of the steps of Algorithm~\ref{algo:subtract} are similar to steps of
  Algorithm~\ref{algo:sum}, except that Hensel lifting is not required here. 
  The complexity analysis is similar and we refer to the proof
  of Proposition~\ref{prop:cost_sum}. The only step which does not appear in
  Algorithm~\ref{algo:sum} is the exact division of $\widehat\chi_1$ by the
  GCD. The cost of this step does not hinder the global complexity since exact
  division of polynomials can be done in quasi-linear
  complexity~\cite[Thm.~9.1]{von2013modern}.
\end{proof}

In practice, if $k$ is sufficiently large, then choosing a global
value for $\lambda$ and using the same value for all the representations of
divisors would succeed with large probability. In this case, we do not need to
call the function~\textsc{ChangePrimElt} within Algorithms~\textsc{AddDivisors}
and \textsc{SubtractDivisors}. This would decrease significantly the
complexities of~\textsc{AddDivisors}
and \textsc{SubtractDivisors}. In any case, this would not change the global
asymptotic complexity of Algorithm~\ref{algo:birdeye}.

\begin{proposition}\label{prop:cost_interp}
  Algorithm~\ref{algo:interpolation} (\textsc{Interpolate}) requires at most
  $O((\deg(\chi)+r)^\omega)$
  arithmetic operations in $k$ and it returns a polynomial of degree less than
  $(\deg(\chi)+r)/\delta+\delta$.
\end{proposition}

\begin{proof}
  First, we recall that $\deg(\chi_E)=r$.
  The computation of the degree $d$ does not cost any arithmetic operations in
  $k$. The construction of the matrix representing the linear map $\varphi$ can
  be done by computing all the modular products $u(S)^iv(S)^j$ modulo $\chi(S)$
  and $\chi_E$ for
  pairs $(i,j)$ such that $i+j\leq d$ and $j<\delta$.
  Lemma~\ref{lem:varphinotinj} states that the number of such pairs is bounded
  above by $3(\deg(\chi)+r)$. By considering the pairs
  $(i,j)$ in increasing lexicographical ordering, computing all these modular
  products can be done within
  $O((\deg(\chi)+r)\Mcomp(\deg(\chi)+r))$ operations in
  $k$. Then, since both dimensions of the matrix are in
  $O(\deg(\chi)+r)$,
  computing a basis of the kernel can be done at cost
  $O((\deg(\chi)+r)^\omega)$ (for instance via a row echelon
  form computation, see~\cite[Thm.~2.10]{storjohann2000algorithms}).

  Next, we show the bound on the degree of the polynomial returned. By
  construction, the inequality $\deg(h) \leq d$ holds so it suffices to show
  that $d < (\deg(\chi)+r)/\delta + \delta$.  
  If $\binom{\delta+1}{2} \leq (\deg(\chi)+r)$, we have $d = \lfloor
  (\deg(\chi)+r)/\delta +
  (\delta-1)/2 \rfloor < (\deg(\chi)+r)/\delta + \delta$. Otherwise, $d = \lfloor
  (\sqrt{1+8(\deg(\chi)+r)}-1)/2 \rfloor < \delta <
  (\deg(\chi)+r)/\delta + \delta$ by direct
  computations. In both cases, we have $\deg(h) < (\deg(\chi)+r)/\delta + \delta$.
\end{proof}

\begin{proposition}\label{prop:cost_princdiv}
  Algorithm~\ref{algo:princdiv} (\textsc{CompPrincDiv}) requires at most $\widetilde O(\max(\deg(q),
  \deg(h))^2\cdot\min(\deg(q), \deg(h)))$ arithmetic operations in $k$.
\end{proposition}
\begin{proof}
  The two costly steps in Algorithm~\ref{algo:princdiv} are the computations of
  the resultant and of the subresultant of two bivariate polynomials. This can be
  done within $\widetilde O(\max(\deg(q),
  \deg(h))^2\cdot\min(\deg(q), \deg(h)))$ operations using
  \cite[Coro.~11.21]{von2013modern}. The Bézout bound implies that the degree of the
  resultant $\widetilde\chi(S)$ is at most $\deg(q)\deg(h)$, hence the
  complexities of all the other steps are quasi-linear in $\deg(q)\deg(h)$, which
  is negligible compared to the cost of the computation of the resultant and
  the subresultant.
\end{proof}

We point out that the complexity of computing resultants and subresultants of
bivariate polynomials have been recently improved
in~\cite{villard2018computing, van2019fast} under some genericity
assumptions. However, since the cost in
Proposition~\ref{prop:cost_princdiv} will be negligible in the global
complexity estimate, we make no effort to optimize it further.

\begin{proposition}\label{prop:num_basis}
  Algorithm~\ref{algo:numbasis} (\textsc{NumeratorBasis}) requires at most
  $O((\deg(\chi)+r)^\omega)$ arithmetic operations in $k$.
\end{proposition}

\begin{proof}
  By Lemma~\ref{lem:varphinotinj}, the domain of the map $\varphi$ has dimension
  $O(\deg(\chi)+r)$.
  Using the monomial basis, the matrix representing the map $\varphi$ can be
  constructed within $\widetilde O((\deg(\chi)+r)^2)$ operations by doing as in
  the proof of Proposition~\ref{prop:cost_interp}. 
  Similarly to the proof of Proposition~\ref{prop:cost_interp}, a basis of the kernel
  of this matrix can be obtained by computing first a row echelon form of the
  matrix within $O( (\deg(\chi)+r)^\omega)$ operations~\cite[Thm.~2.10]{storjohann2000algorithms}.
\end{proof}

\begin{table}
 \begin{adjustbox}{center}
  \begin{tabular}[t]{|c|l|}
\hline
Divisor&Degree\\\hline
$D_h$&$<\deg(\C)^2 + \deg(D_+)$\\
$D_{\rm res}$&$<\deg(\C)^2$\\
$D_{\rm num}$&$<\deg(\C)^2+\deg(D_+)$\\
\hline
\end{tabular}
\begin{tabular}[t]{|c|l|}
\hline
  Subroutine & Complexity\\
  \hline
  \textsc{Interpolate}&$O( (\deg(D_+)+r)^\omega)$\\
  \textsc{CompPrincDiv}&$\widetilde O(\max(\deg(\C)^3,(\deg(D_+)+r)^2/\deg(\C)))$\\
  \textsc{SubtractDivisors}&$O(\max(\deg(\C)^{2\omega}, (\deg(D_+))^\omega))$\\
  \textsc{AddDivisors}&$O(\max(\deg(\C)^{2\omega}, (\deg(D_+))^\omega))$\\
  \textsc{NumeratorBasis}&$O(\max(\deg(\C)^{2\omega}, \deg(D_+)^\omega))$\\
  \hline
\end{tabular}
 \end{adjustbox}
\caption{Degrees of divisors and complexities of the subroutines in terms of
  the input size.\label{table:glob_compl}}
\end{table}

All the complexities and the degree estimates computed in this section are summed up in
Table~\ref{table:glob_compl}. For bounding the degree of $D_{\sf num}=D_{\sf
res}+D_-$, we use the fact that we can assume without loss of generality that
$\deg(D_-)\leq \deg(D_+)$, since otherwise $L(D_+-D_-)$ is reduced to
$0$. Summing all the complexity bounds yields the global
complexity bound:

\begin{theorem}\label{thm:compl}
  Algorithm~\ref{algo:birdeye} (\textsc{RiemannRochBasis}) requires at most
  $O(\max(\deg(\C)^{2\omega}, \deg(D_+)^\omega))$ arithmetic operations in $k$.
\end{theorem}

\begin{proof}
  A direct consequence of Propositions~\ref{prop:cost_sum},
  \ref{prop:cost_sub}, \ref{prop:cost_interp}, \ref{prop:cost_princdiv} and
  \ref{prop:num_basis} is that the complexity of Algorithm~\ref{algo:birdeye}
  is bounded by $O(\max(\deg(\C)^{2\omega}, (\deg(D_+)+r)^\omega))$. The proof
  is concluded by noticing that $r=O(\deg(\C)^2)$ since
  $g=\binom{\deg(\C)-1}2-r$ is nonnegative.
\end{proof}

\section{Lower bounds on the probability of success}\label{sec:proba}

In this section, we examinate all possible sources of failures for the main
algorithm. In fact, if the assumptions detailed in Section~\ref{sec:overview}
are satisfied, then failure can only come from a bad choice of an
element picked at random. More precisely, we show that these bad choices
can be characterized algebraically and that they are included in the set of
roots of polynomials. Bounding the degrees of these polynomials provides us
with lower bounds on the probability of success if random
elements in $k$ are picked uniformly at random in a finite subset $\mathcal E\subset
k$.

First, we investigate which values of $\lambda$ make Algorithm~\ref{algo:changeprim}
  (\textsc{ChangePrimElt}) fail: These are the values of $\lambda$ such that
  there is a line of equation $\lambda X+Y+\gamma$ for some $\gamma\in\overline
  k$ which goes either through two distinct points in the support of the input
  divisor, or which is tangent to $\C^0$ at a point in
  the support of the divisor.

\begin{proposition}\label{prop:degchangeprim}
  Given an effective divisor $D=(\lambda, \chi, u, v)$, the set
  of $\widetilde\lambda\in k$ such that Algorithm~\ref{algo:changeprim}
  (\textsc{ChangePrimElt}) with
  input $D,\widetilde\lambda$ fails is contained in the set of roots of a
  nonzero univariate polynomial with coefficients in $k$ of degree at most
  $\binom{\deg(\chi)+1}2$.
\end{proposition}

\begin{proof}There are two possible sources of failures for
  Algorithm~\ref{algo:changeprim}: if the vector $(1, -\widetilde\lambda)$ is
  tangent to the curve $\C^0$ at one of the points in the support of the
  effective divisor (first test) or if $\widetilde\lambda X+Y$ is not a
  primitive element (second test).

  Let $J=\langle\chi(\lambda X+Y), X-u(\lambda X+Y), Y-v(\lambda
  X+Y)\rangle\subset k[\C^0]$ be the ideal associated to the effective divisor
  $D$, and let $\Delta$ be the polynomial given by
  Lemma~\ref{lem:existence_primrepr_aux} for $J$.
  By construction, the polynomial $\Delta$ satisfies the wanted properties.
\end{proof}

Before investigating Algorithms~\ref{algo:sum} and~\ref{algo:subtract}
  (\textsc{AddDivisors} and
  \textsc{SubtractDivisors}), we need a technical lemma.

\begin{lemma}\label{lem:primeltproj}
Let $\phi: R\rightarrow S$ be a surjective morphism of finite $k$-algebras, and let $z$ be a primitive
element for $R$. Then $\phi(z)$ is a primitive element for $S$.
\end{lemma}

\begin{proof}
Since $\phi$ is surjective, any element $y\in S$ equals $\phi(x)$ for some
$x\in R$. Since $z$ is primitive, there exists a univariate polynomial $w(S)\in
k[S]$ such that $x=w(z)$. Consequently, $y=\phi(w(z))=w(\phi(z))$.
Therefore, $\phi(z)$ is primitive for $S$.
\end{proof}

\begin{proposition}\label{prop:algsumdegbounds}
  For a given input $(q, D_1, D_2)$ of Algorithm~\ref{algo:sum}
  (\textsc{AddDivisors}), the set of $\widehat\lambda$
  which makes Algorithm~\ref{algo:sum} fail is contained in the set of roots of a
  nonzero univariate polynomial with coefficients in $k$ and of degree bounded
  by 
$\binom{\deg(\chi_1)+\deg(\chi_2)+1}{2}$.
\end{proposition}

\begin{proof}
For $i\in \{1, 2\}$, consider $I_i=\langle U-u_i(S), V- v_i(S),
\chi_i(S)\rangle\cap k[U,V]$.
Lemma~\ref{lem:existence_primrepr_aux} for the ideal $I_1\cdot I_2+\langle
q(U,V)\rangle$ yields a
nonzero polynomial $\Delta$ of degree
at most $\binom{\deg(\chi_1)+\deg(\chi_2)+1}{2}$. We will prove that this
polynomial satisfies the wanted properties.

By definition, elements $\widehat\lambda\in k$ which are not roots of $\Delta$ are such
that $\widehat\lambda U + V$ is a primitive element for
$\red(k[U,V]/(I_1\cdot I_2))$ and
\begin{equation}\label{eq:cond_tangent_sum}
  \begin{array}{rcl}
\displaystyle\GCD\left(\frac{\partial q}{\partial X}(u_1(S),v_1(S))-\widehat\lambda
\frac{\partial q}{\partial Y}(u_1(S),
    v_1(S)),\chi_1(S)\right) &=& 1, \\
\displaystyle\GCD\left(\frac{\partial q}{\partial X}(u_2(S),v_2(S))-\widehat\lambda
\frac{\partial q}{\partial Y}(u_2(S),
    v_2(S)),\chi_2(S)\right) &=& 1.
  \end{array}
\end{equation}
There are three possible sources of failure for Algorithm~\ref{algo:sum}: the
two calls to \textsc{ChangePrimElt}, and the conditional test. 
The fact that the calls to \textsc{ChangePrimElt} succeed is a direct
consequence of Lemma~\ref{lem:primeltproj}, using the canonical projections
$\red(k[U,V]/(I_1\cdot I_2))\rightarrow\red(k[U,V]/I_i)$ for $i\in\{1,2\}$, see
the proof of Proposition~\ref{prop:degchangeprim}.
By
Lemma~\ref{lem:techrepr} and Proposition~\ref{prop:correct_changeprim}, we have that for $i\in\{1,2\}$, $I_i = \langle U-\widehat u_i(S), V- \widehat v_i(S),
\widehat\chi_i(S)\rangle\cap k[U,V]$. 
Then $\widehat\lambda U+V$ must be a
primitive element for $\red(k[U,V]/I_1)$ and for $\red(k[U,V]/I_2)$ by
Lemma~\ref{lem:divHprimelt}. Let $\widetilde\chi_1$
and $\widetilde\chi_2$ denote the minimal polynomials of $\widehat\lambda U+V$
in $\red(k[U,V]/I_1)$ and $\red(k[U,V]/I_2)$. Also, set
$\xi=\LCM(\widetilde\chi_1, \widetilde\chi_2)$. Consequently, $\widetilde\chi_1(\widehat\lambda U+V)\cdot
\widetilde\chi_2(\widehat\lambda U+V)\in \sqrt{I_1\cdot I_2}$ and
$\xi$ is the minimal polynomial of $\widehat\lambda U+V$ in
$\red(k[U,V]/(I_1\cdot I_2))=k[U,V]/(\sqrt{I_1}\cap \sqrt{I_2})$. Since
$\widehat\lambda U+V$ is a primitive element for $\red\left( k[U,V]/(I_1\cdot
I_2)\right)$, then 
the canonical map $$\red\left(k[U,V]/(I_1\cdot I_2)\right)\rightarrow\red(
k[U,V]/I_2)\times \red(k[U,V]/I_1)$$ becomes a map
$$k[S]/\xi(S)\rightarrow k[S]/\widetilde\chi_1(S)\times k[S]/\widetilde\chi_2(S).$$ This
implies that there exists an element $\widetilde u_{12}\in k[S]/\xi(S)$ (which
is in fact the class of $U$ in $\red(k[U,V]/(I_1\cdot I_2))$) such
that $\widetilde u_{12}\equiv \widehat u_1\bmod \widetilde\chi_1$ and
$\widetilde
u_{12}\equiv \widehat u_2\bmod \widetilde\chi_2$. As a consequence, $\widehat
u_1\equiv \widehat u_2\bmod\GCD(\widetilde\chi_1, \widetilde\chi_2)$. Using Hensel's
lemma, the property \textbf{(Div-H3)} and the CRT, we obtain that the
equation $q(u(S), S-\lambda u(S))=0$ has a unique solution $s$ in
$k[S]/\GCD(\widehat\chi_1(S),
\widehat\chi_2(S))$ such that $s\equiv \widehat
u_1\equiv \widehat u_2\bmod\GCD(\widetilde\chi_1, \widetilde\chi_2)$. By \textbf{(Div-H1)}, both $\widehat
u_1$ and $\widehat u_2$ are solutions, and therefore 
$\widehat u_1\equiv \widehat u_2\bmod\GCD(\widehat\chi_1,
\widehat\chi_2)$, which shows that the last conditional test succeeds.

\end{proof}

\begin{proposition}\label{prop:degbound_subtract}
  For a given input $(q, D_1, D_2)$ of Algorithm~\ref{algo:subtract}
  (\textsc{SubtractDivisors}), the set of $\widehat\lambda$
  which makes Algorithm~\ref{algo:subtract} fail is contained in the set of roots of a
  nonzero univariate polynomial with coefficients in $k$ and of degree bounded
  by 
$\binom{\deg(\chi_1)+\deg(\chi_2)+1}{2}$.
\end{proposition}

\begin{proof}
  The proof is similar to the first part of the proof of
  Proposition~\ref{prop:algsumdegbounds}. With the same notation as in the
  proof of Proposition~\ref{prop:algsumdegbounds},
  Algorithm~\ref{algo:subtract} fails only if $\widehat\lambda U+V$ is not a
  primitive element for $\red(k[U,V]/(I_1\cdot I_2))$ or if the vector
  $(1, -\widetilde\lambda)$ is tangent to the curve at one of the points in the
  support of one of the divisors. Using a proof similar to that of 
  Proposition~\ref{prop:algsumdegbounds}, this happens only when
  $\widehat\lambda$ is in the set of roots of the nonzero polynomial of degree at
  most
$\binom{\deg(\chi_1)+\deg(\chi_2)+1}{2}$ provided by
Lemma~\ref{lem:existence_primrepr_aux} for the ideal $I_1\cdot I_2+\langle
q(U,V)\rangle$.
\end{proof}

Next, we wish to bound the probability that
Algorithm~\ref{algo:princdiv} (\textsc{CompPrincDiv}) fails.
Before stating the next proposition, we recall the second assumption that we
have made on the input divisor and which is described in
Section~\ref{sec:overview}. It ensures the existence of a form $h\in\overline
k[\C]$ of given
degree such that $(h)\geq D_++E$ and $(h)-E$ is smooth.
With the notation in the following proposition, this assumption precisely means
that $A\ne\ker(\varphi)\otimes_k \overline k$.

\begin{proposition}\label{prop:badjoin}
  Let $A\subset \ker(\varphi)\otimes_k \overline k \subset \overline k[X, Y]$ be the subset of all the regular functions $h$ in the kernel of $\varphi$ in
Algorithm~\ref{algo:interpolation} which are such that $D_h =(h)-E$
is not a smooth divisor. If $A\ne \ker(\varphi)$, then $A$ is contained in the join of at most
$2r$ hyperplanes in $\ker(\varphi)\otimes_k \overline k$. Consequently, there
is a nonzero polynomial in $\overline k[Z_1,\ldots, Z_{\dim(\ker(\varphi))}]$ of degree at most
$2r$ which vanishes at values
$(\mu_1,\ldots,\mu_{\dim(\ker(\varphi))})$ for which the third test in Algorithm~\ref{algo:princdiv}
fails for all $\lambda\in\overline k$.
\end{proposition}

\begin{proof}
  If $D_h$ involves a point $P$ of $\widetilde C$ which projects to a node, then $(h)\geq E+P$. 
  The set of regular functions $h$ of a given degree which satisfy $(h)\geq
  E+P$ is a linear space. The set of such $h$ in
  $\ker(\varphi)\otimes_k\overline k$ forms a proper
  subspace since $A\ne \ker(\varphi)\otimes_k\overline k$.
  Consequently, it is contained in an hyperplane. 
 Such a hyperplane $H$ can be described
  by a linear form $\psi$ in $\overline k[Z_1,\ldots, Z_{\dim(\ker(\varphi))}]$ such that
  $\psi(\mu_1,\ldots, \mu_{\dim(\ker(\varphi))})=0$ if and only if
  $\sum_{i=1}^{\dim(\ker(\varphi))} \mu_i\mathbf
  b_i\in H$, where $\mathbf b_1,\ldots, \mathbf b_{\dim(\ker(\varphi))}$ is a
  basis of $\ker(\varphi)$.  

  Iterating this argument over all the $2r$ points of the nonsingular model $\widetilde \C$ which
  project to nodes, we obtain that $A$ is contained in the join of $2r$
  hyperspaces. Multiplying the $2r$ corresponding linear forms in $\overline k[Z_1,\ldots, Z_{\dim(\ker(\varphi))}]$
  proves the last sentence of the
  proposition.
\end{proof}

\begin{proposition}\label{prop:deg_compprinc_bis}
  The set of values of $\lambda$ which make the first test in
Algorithm~\ref{algo:princdiv} fail is contained within the set of roots of a
  nonzero univariate polynomial with coefficients in $k$ of degree $\deg(\C)+1.$
\end{proposition}
\begin{proof}
 Writing $\widetilde q(S,Y) = q((S-Y)/\lambda, Y)$, the first test fails if
 $\lambda=0$ or if the coefficient of the monomial $Y^{\deg(\C)}$ in
 $\widetilde q$
 vanishes. Writing explicitly the change of variables, we obtain that this
 coefficient equals $\sum_{i=0}^{\deg(\C)} (-1/\lambda)^i q_{i, \deg(\C)-i}$,
 where $q_{i,j}$ stands for the coefficient of $X^i Y^j$ in $q$. Multiplying by
 $\lambda^{\deg(\C)+1}$ clears the denominator and adds the root $0$ to exclude
 the case $\lambda = 0$; this provides a polynomial satisfying the desired
 properties.
\end{proof}

\begin{proposition}\label{prop:deg_compprinc}
  Let $h\in k[\C^0]$ be a regular function such that the support of $(h)-E$
  does not contain any singular point. Then the set of $\lambda$
which makes Algorithm~\ref{algo:princdiv} (\textsc{CompPrincDiv}) with input $q,
  h$ fail is contained in the set of roots of a
  nonzero univariate polynomial with coefficients in $k$ and of degree bounded
  by $2\binom{\deg(\C)\deg(h)+1}2+2r+\deg(\C)+1$.
\end{proposition}

\begin{proof}
First, let $\Delta_1$ be the univariate polynomial constructed in
Proposition~\ref{prop:deg_compprinc_bis}. The first test in Algorithm~\ref{algo:princdiv} does
not fail only if $\lambda$ is not a root of $\Delta_1$.

The second test in Algorithm~\ref{algo:princdiv} fails only if $\lambda$ is a root of $T_E$.

  By Bézout theorem, the effective divisor $(h)$ has degree at most
  $\deg(\C)\deg(h)$. 
  Therefore, Lemma~\ref{lem:primeltbidegbound} for the ideal $\sqrt{\langle
  q,h\rangle}$
  yields a nonzero polynomial $\Delta_2$ of degree at most
  $\binom{\deg(\C)\deg(h)}2$ such that the set of
  $\lambda$ such that $\lambda X+Y$ is not a primitive element for
  $\red(k[\C^0]/\langle h\rangle)$.
  Since $(h)\geq E$, the fact that $\Delta_2(\lambda)\ne 0$ implies
  that $\lambda X+Y$ is a primitive element for the $k$-algebra associated to the nodal divisor, and hence the call to the function \textsc{ChangePrimEltNodal} in
  Algorithm~\ref{algo:princdiv} does not fail. Since by assumption $(h)-E$ is
  smooth, this also implies that the roots of $\widehat\chi_E$ are roots of
  $\chi$ with multiplicity exactly $2$ by Lemma~\ref{lem:tech_res}. Consequently, if $\lambda$ is
  not a root of $\Delta_2$, then $\GCD(\chi, \widehat\chi_E)=1$ and
  therefore the third test in Algorithm~\ref{algo:princdiv} must succeed.

  Finally, Lemma~\ref{lem:existence_primrepr_aux} for the ideal
  $\sqrt{\langle
  q,h\rangle}:I_E^\infty\subset k[\C^0]$
  yields a nonzero polynomial $\Delta_3$ of degree at most
  $\binom{\deg(\C)\deg(h)+1}2$ such that the set of
  $\lambda$ such that $\lambda X+Y$ is not a primitive element for
  $\red(k[\C^0]/\langle h\rangle)$ or such that the last test in Algorithm~\ref{algo:princdiv} fails is
  contained within the set of roots of $\Delta_3$.

 We claim that the product $\Delta_1\cdot \Delta_2\cdot\Delta_3\cdot T_E$ satisfies the required
 properties. To prove this claim, it remains to show that if $\lambda$ is not a root 
 of $\Delta_1\cdot \Delta_2\cdot\Delta_3\cdot T_E$, then the fourth test
 succeeds, i.e. $a_1(S)$ is invertible
 modulo $\chi(S)$.

 To this end, we notice that $a_1(S)$ is invertible modulo $\chi(S)$ if and
 only if
 $a_1(s)$ is nonzero for any root $s\in \overline k$ of $\chi(S)$. By
 \cite[Cor.~5.1]{kahoui2003elementary}, this is equivalent to the fact that the GCD
 of the polynomials $q((s-Y)/\lambda, Y)$, $h( (s-Y)/\lambda, Y)$ has degree
 $1$ for any root $s$ of $\chi(S)$.  Next, we note that if $\lambda$ is not a
 root of $\Delta_3$, then any common root $y$ of $q((s-Y)/\lambda, Y)$ and
 $h( (s-Y)/\lambda, Y)$ has multiplicity $1$ in $q((s-Y)/\lambda, Y)$:
 Indeed, the vanishing of the derivative $\partial/\partial Y$ of
 $q((s-Y)/\lambda, Y)$ at $Y=y$ would precisely mean that the vector
 $(1,-\lambda)$ is tangent to the curve at the intersection point, which is
 impossible by definition of $\Delta_3$.  Consequently, the GCD of the
 polynomials $q((s-Y)/\lambda, Y)$, $h( (s-Y)/\lambda, Y)$ must be squarefree.
 Finally, let $y_1, y_2\in \overline k$ be two common roots of
 $q((s-Y)/\lambda, Y)$, $h( (s-Y)/\lambda, Y)$. This means that $(
 (s-y_1)/\lambda, y_1)$ and $( (s-y_2)/\lambda, y_2)$ are two common zeros of
 $q(X,Y)$ and $h(X,Y)$. Since $\lambda$ is not a root of $\Delta_2$, $\lambda
 X+Y$ is a primitive element for $\red(k[\C^0]/\langle h\rangle)=k[X,Y]/\langle
 q, h\rangle$, which implies that $\lambda X+Y$ takes distinct values at all
 points $(x,y)$ in the variety associated to the system $h(X,Y)=q(X,Y)=0$
 (the endomorphism of multiplication by $\lambda X+Y$ must have
 distinct eigenvalues, see e.g.  the proof of Lemma~\ref{lem:primeltbidegbound}
 for more details). In particular, this means that $y_1=y_2$, since
 $\lambda X+Y$ takes the same value $s$ at
 $((s-y_1)/\lambda, y_1)$ and $((s-y_2)/\lambda, y_2)$.  Consequently,
 the GCD of the polynomials $q((s-Y)/\lambda, Y)$, $h( (s-Y)/\lambda, Y)$ is a
 squarefree polynomial with at most one root, hence it has degree at most $1$.
 Since $\Resultant(q((s-Y)/\lambda, Y), h(
 (s-Y)/\lambda, Y))$ vanishes and the coefficient of $Y^{\deg(q)}$ in
 $q((S-Y)/\lambda, Y)$ is nonzero because
 $\lambda$ is not a root of $\Delta_1$, this GCD must have degree at least $1$. Therefore,
 this GCD has degree exactly $1$, and hence $a_1(S)$ is invertible modulo
 $\chi(S)$.
 \end{proof}

 \begin{table}
 \begin{adjustbox}{center}
  \begin{tabular}{|c|c|c|}
    \hline
    Algorithm & Failure probability & Statement\\
    \hline
    \hline
    \textsc{ChangePrimElt}&$\deg(D)^2/\lvert \mathcal E\rvert$&Prop.~\ref{prop:degchangeprim}\\
    \hline
    \textsc{AddDivisors}&$O(\max(\deg(D_1),\deg(D_2))^2/\lvert \mathcal E\rvert)$&Prop.~\ref{prop:algsumdegbounds}\\
    \hline
    \textsc{SubtractDivisors}&$O(\max(\deg(D_1),\deg(D_2))^2/\lvert \mathcal E\rvert)$&Prop.~\ref{prop:degbound_subtract}\\
    \hline
    \textsc{CompPrincDiv}&$O(\deg(\C)^2\deg(h)^2/\lvert
    \mathcal E\rvert)$&\begin{tabular}{c}Prop.~\ref{prop:badjoin}\\
      Prop.~\ref{prop:deg_compprinc}\\
      Schwartz-Zippel
      lemma~\cite[Coro.~1]{schwartz1979probabilistic}
  \end{tabular}\\
    \hline
\end{tabular}
\end{adjustbox}
\caption{Probabilities of failure.\label{table:probas}}
\end{table} 

Finally, we can derive our bound on the probability that the toplevel
algorithm fails by summing the probabilities that the subroutines fail.

\begin{theorem}\label{thm:glob_bound_proba}
  Let $\mathcal E\subset k$ be a finite set.
  Assume that each call to the function \textsc{Random}($k$) is done by picking an
  element uniformly at random in $\mathcal E$. Then the probability that
  Algorithm~\ref{algo:birdeye} fails is bounded above by
  $$O(\max(\deg(\C)^4, \deg(D_+)^2)/\lvert \mathcal E\rvert).$$
\end{theorem}

\begin{proof}
 Propositions~\ref{prop:algsumdegbounds},
 \ref{prop:degbound_subtract}, together with the fact that the number of
 roots in $k$
 of a univariate polynomial is bounded by its degree, directly imply that the
 probabilities of failure of Algorithms~\textsc{AddDivisors} and \textsc{SubtractDivisors} are bounded by
 $O(\max(\deg(D_1),\deg(D_2))^2/\lvert \mathcal E\rvert)$, if the computation of the
 characteristic polynomial in Algorithm~\textsc{ChangePrimElt} succeeds.
 Following~\cite{pernet2007faster} (see also the remark in the proof of
 Proposition~\ref{prop:cost_changeprim}), the probability that the computation of
 the characteristic polynomial in \textsc{ChangePrim} fails is bounded by
 $\deg(\chi)^2/\lvert \mathcal E\rvert$. Therefore, the probabilities that
 Algorithms~\textsc{AddDivisors} and \textsc{SubtractDivisors} fail are still
 bounded by $O(\max(\deg(D_1,D_2))^2/\lvert \mathcal E\rvert)$ when we take into account
 the probability that the computations of the characteristic polynomials fail.
 Notice that our second technical assumption (described in
 Section~\ref{sec:overview}) on the input divisor ensures that $A\ne\ker(\varphi)$ in
 Proposition~\ref{prop:badjoin}.
Using Proposition~\ref{prop:badjoin}, Schwartz-Zippel lemma~\cite[Coro.~1]{schwartz1979probabilistic},
 Proposition~\ref{prop:deg_compprinc}, together with the fact that $r\leq
 \binom{\deg(\C)-1}2$, we bound the probability that
 \textsc{CompPrincDiv} fails by
 $O(\deg(\C)^2\deg(h)^2/\lvert \mathcal E\rvert)$. 
 
 The failure probabilities are summed up
 in Table~\ref{table:probas}. Next, notice that the probability of failure of
 Algorithm~\ref{algo:birdeye} is bounded by the sum of the probabilities of the
 subroutines. Finally, the proof is concluded by using the inequality
 $\deg(h)<(\deg(D_+)+r)/\deg(\C)+\deg(\C)$ (Proposition~\ref{prop:cost_interp}) and the degree bounds in Table~\ref{table:glob_compl} for the divisors
 arising in Algorithm~\ref{algo:birdeye}.
\end{proof}

\paragraph{Deciding whether the assumptions on the input divisor are
satisfied.} The result in
Theorem~\ref{thm:glob_bound_proba} only holds true if the assumptions on the
input divisor described in Section~\ref{sec:overview} are satisfied. The first
assumption --- namely, the smoothness of the input divisor --- can be easily checked, so we focus here on deciding whether the
second assumption is satisfied or not. Namely, this assumption
requires the
existence of a form $h\in \overline k[\C]$ of degree $d$ --- where $d$ is the value computed during the
execution of Algorithm~\textsc{Interpolate} --- such that $(h)\geq D_++E$ and 
$(h)-E$ is smooth. In order to have a complete Las
Vegas algorithm, we need to be able to check whether this condition is
satisfied.
To this end, instead of returning
only one form during Algorithm~\textsc{Interpolate}, we can return a basis
$(h_1,\ldots,h_\ell)$ of the
forms $h$ such that $(h)\geq D_++E$. Then, we check if there exists a point
above a node which is simultaneously in the support of all the divisors
$\{(h_i)-E\}_{i\in \{1,\ldots,\ell\}}$. This boils
down to computing primitive element representations of the principal divisors
$(h_1),\ldots,(h_\ell)$, which is done by running on these $\ell$ forms a modified version of Algorithm~\textsc{CompPrincDiv} where the
last test is removed in order to allow singular points. 
Each execution of Algorithm~\textsc{CompPrincDiv} costs $\widetilde
O(\max(\deg(\C)^3, (\deg(D_+)+r)^2/\deg(\C)))$ operations in $k$, and thus ---
using the fact that $\ell=O(\deg(D_+)+\deg(\C)^2)$ --- the total cost of the
procedure is bounded above by $\widetilde O(\max(\deg(\C)^5,
\deg(D_+)^{5/2}))$. Therefore, in theory, running this decision procedure
increases the overall complexity stated in Theorem~\ref{thm:compl} since the best known value of $\omega$ is less than
$5/2$. However, in practice this does not change the asymptotic complexity since
practical algorithms for linear algebra rely on Gauss or Strassen
approaches; In this case, $\omega > 5/2$, and hence the cost of this
verification procedure is negligible compared to the global complexity of our
algorithm. Multiplying the probability of failure of
Algorithm~\textsc{CompPrincDiv} by the number of basis vectors yields the
bound $O(\max(\deg(D_+)^3,\deg(\C)^6)/\lvert\mathcal E\rvert)$
for the probability of failure of this verification procedure.

\section{Experimental results}\label{sec:expe}

We have implemented Algorithm~\ref{algo:birdeye} in C++ for $k=\mathbb
Z/p\mathbb Z$, relying on the NTL library
for all operations on univariate polynomials and for linear algebra. 
We have also implemented the group law on the Jacobian of a curve via Riemann-Roch space computations.
Our software \texttt{rrspace} is freely available at
\url{https://gitlab.inria.fr/pspaenle/rrspace} and it is distributed under the
LGPL-2.1+ license.

All the
experiments presented below have been conducted on a
Intel(R)~Core(TM)~i5-6500~CPU@3.20GHz with 16GB RAM. The comparisons with the
computer algebra system Magma have been done with its version V2.23-8.

\pgfplotstableread{rrspace1.data}{\rrspaceA}
\pgfplotstableread{rrspace2.data}{\rrspaceB}
\pgfplotstableread{rrspace3.data}{\rrspaceC}
\pgfplotstableread{rrspace4.data}{\rrspaceD}
\pgfplotstableread{magma1.data}{\magmaA}
\pgfplotstableread{magma2.data}{\magmaB}
\pgfplotstableread{magma3.data}{\magmaC}
\pgfplotstableread{magma4.data}{\magmaD}

\begin{figure}
 \begin{adjustbox}{center}
\begin{tikzpicture}[scale=0.9]
  \pgfplotsset{every axis legend/.append style={
at={(0.05,0.95)},
anchor=north west}}
\begin{loglogaxis}[minor tick num=1,
xlabel=Degree of the divisor,
ylabel=Time in seconds,
xtick pos=left,
ytick pos=left,
width=8cm,
height=6cm,
log basis x=2,
log basis y=2]
\addplot [blue,very thick] table [x={d}, y={t}] {\magmaA};
\addlegendentry{\large\texttt{Magma}}
\addplot [red,very thick] table [x={d}, y={t}] {\rrspaceA};
\addlegendentry{\large\texttt{rrspace}}
\end{loglogaxis}
\end{tikzpicture}
\begin{tikzpicture}[scale=0.9]
  \pgfplotsset{every axis legend/.append style={
at={(0.05,0.95)},
anchor=north west}}
\begin{loglogaxis}[minor tick num=1,
xlabel=Degree of the divisor,
width=8cm,
height=6cm,
xtick pos=left,
ytick pos=left,
log basis x=2,
log basis y=2]
\addplot [blue,very thick] table [x={d}, y={t}] {\magmaB};
\addlegendentry{\large\texttt{Magma}}
\addplot [red,very thick] table [x={d}, y={t}] {\rrspaceB};
\addlegendentry{\large\texttt{rrspace}}
\end{loglogaxis}
\end{tikzpicture}
\end{adjustbox}
\caption{Comparison of the time required by \texttt{rrspace} and \texttt{Magma}
to compute a basis of $L(D)$ on a fixed smooth curve of degree $10$ over
$\mathbb Z/65521\mathbb Z$. On the left, $D$ is the sum of random irreducible
effective divisors of degree $10$. On the right, $D$ is a multiple of an
irreductible divisor of degree $10$. Both axes are in logarithmic
scale.\label{fig:expe_rr1}}
\end{figure}
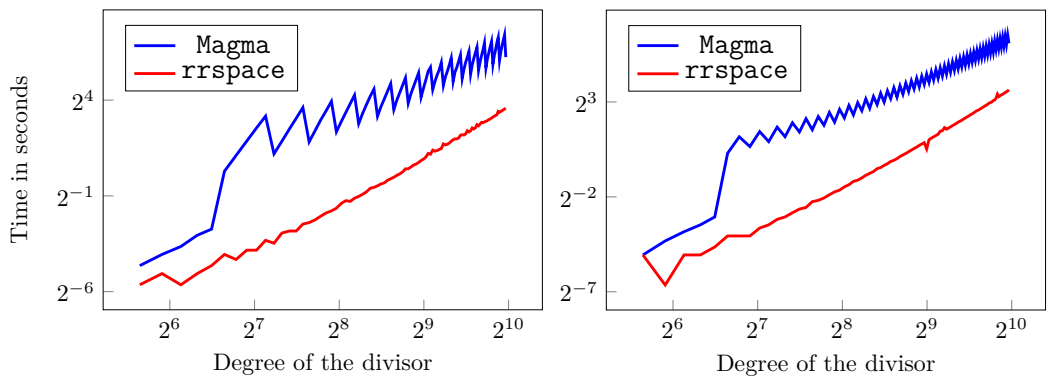

\begin{figure}
 \begin{adjustbox}{center}
\begin{tikzpicture}[scale=0.9]
  \pgfplotsset{every axis legend/.append style={
at={(0.05,0.95)},
anchor=north west}}
\begin{loglogaxis}[minor tick num=1,
xlabel=Degree of the divisor,
ylabel=Time in seconds,
xtick pos=left,
ytick pos=left,
width=8cm,
height=6cm,
log basis x=2,
log basis y=2]
\addplot [blue,very thick] table [x={d}, y={t}] {\magmaC};
\addlegendentry{\large\texttt{Magma}}
\addplot [red,very thick] table [x={d}, y={t}] {\rrspaceC};
\addlegendentry{\large\texttt{rrspace}}
\end{loglogaxis}
\end{tikzpicture}
\begin{tikzpicture}[scale=0.9]
  \pgfplotsset{every axis legend/.append style={
at={(0.05,0.95)},
anchor=north west}}
\begin{loglogaxis}[minor tick num=1,
xlabel=Degree of the divisor,
xtick pos=left,
ytick pos=left,
width=8cm,
height=6cm,
log basis x=2,
log basis y=2]
\addplot [blue,very thick] table [x={d}, y={t}] {\magmaD};
\addlegendentry{\large\texttt{Magma}}
\addplot [red,very thick] table [x={d}, y={t}] {\rrspaceD};
\addlegendentry{\large\texttt{rrspace}}
\end{loglogaxis}
\end{tikzpicture}
\end{adjustbox}
\caption{Comparison of the time required by \texttt{rrspace} and \texttt{Magma}
to compute a basis of $L(D)$ on a fixed curve of degree $10$, where $D$ is the sum of random irreducible
effective divisors of degree $10$. On the left, the base field is $\mathbb
Z/65521\mathbb Z$ and the curve is nodal. On the right, the base field is
$\mathbb Z/(2^{32}-5)\mathbb Z$ and the curve is smooth. Both axes are in
logarithmic scale.\label{fig:expe_rr2}}
\end{figure}
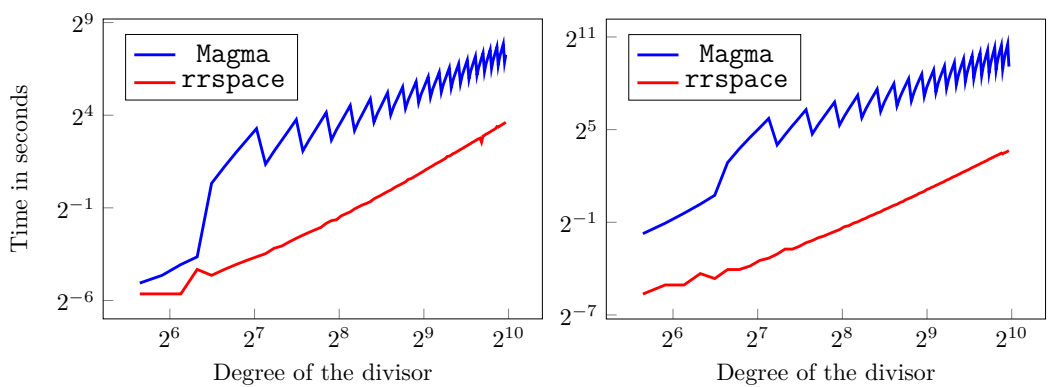

Our first experimental data is generated as follows. We set $k=\mathbb Z/65521\mathbb Z$. For $i$
from 10 to 100, we
consider a curve $\C$ defined by a random bivariate polynomial of degree $10$ over
$k$, and we generate $i$ random irreducible $k$-defined effective divisors
$D_1,\ldots, D_i$ of
degree $10$ on $\C$ by using the
\texttt{RandomPlace()} function in \texttt{Magma}. Then we set $D=D_1+\dots+D_i$ and we measure the time used for
computing a basis of $L(D)$ by using
either \texttt{Magma} via its function
  \texttt{RiemannRochSpace()} or
the software \texttt{rrspace}. The experimental results are displayed in the
left part of
Figure~\ref{fig:expe_rr1}. For these parameters, we observe that
\texttt{rrspace} has a speed-up larger than $7$ compared to
\texttt{Magma}. Since we do not have access to the implementation of the
function \texttt{RiemannRochSpace()} in \texttt{Magma}, we cannot explain the
small variations which appear in the \texttt{Magma} timings.

Our second experimental data investigate the behavior of our algorithm when the input
divisor contains multiplicities. To this end, we generate the input divisor as
a multiple of a random place of degree $10$ on the curve.
The experimental results are displayed in the
right part of
Figure~\ref{fig:expe_rr1}. For these parameters, we observe that
\texttt{rrspace} has a speed-up larger than $6$ compared to
\texttt{Magma}.

Our third experimental data study the behavior of our algorithm in the
presence of nodes. To this end, we fix the following nodal curve defined by the
equation 
$$Q(X,Y,Z) = -Y^2 Z^8 + X^2 Z^8 + Y^4 Z^6 -X^3 Z^7+X^{10}-5\,Y^{10}+3\,X^3 Y^7$$
which has a node at the origin and we generate input divisors as for the
first experimental data. The experimental results are displayed in the
left part of
Figure~\ref{fig:expe_rr2}. For these parameters, we observe that
\texttt{rrspace} has a speed-up larger than $10$ compared to
\texttt{Magma}.

Finally, since the timings are very sensible to the efficiency of the linear
algebra routines, we study what happens for larger finite fields. The fourth
experimental data are generated as for our first experimental
data, but we replace the field $\mathbb Z/65521\mathbb Z$ by the field $\mathbb
Z/(2^{32}-5)\mathbb Z$. Here, the size of the field is out of the range of the
highly optimized arithmetic in \texttt{Magma} for small finite fields, and
consequently we observe speedups larger than $45$ (the speedup goes up to more than $200$ for some examples).
The experimental results are displayed in the
right part of
Figure~\ref{fig:expe_rr2}. 

\bibliographystyle{abbrv}
\bibliography{./biblio}

\bigskip

\footnotesize
\noindent {\bf Authors' addresses:}

\noindent Aude Le Gluher, CARAMBA project, Universit\'e de Lorraine;
Inria Nancy -- Grand Est; CNRS, UMR 7503;
LORIA, Nancy, France, {\tt aude.le-gluher@loria.fr}

\smallskip

\noindent Pierre-Jean Spaenlehauer, CARAMBA project, INRIA Nancy -- Grand Est; Universit\'e de Lorraine; CNRS, UMR 7503; 
LORIA, Nancy, France, {\tt pierre-jean.spaenlehauer@inria.fr}

\end{document}